\documentclass[12pt]{article}

\usepackage[english]{babel}
\usepackage{bm}
\usepackage{amsmath}
\usepackage[table]{xcolor}
\usepackage{amsthm}
\usepackage{amsfonts}
\usepackage{amssymb}
\usepackage{amscd}
\usepackage{graphicx}
\usepackage[authoryear]{natbib}
\usepackage{float}
\usepackage{wrapfig}
\usepackage{enumitem}
\usepackage{subcaption}

\newcommand{\blind}{1}

\usepackage[margin=1.0in]{geometry}

\def\spacingset#1{\renewcommand{\baselinestretch}%
{#1}\small\normalsize} \spacingset{1}

\theoremstyle{plain}

\newtheorem{theorem}{Theorem}[section]
\newtheorem{lemma}[theorem]{Lemma}
\newtheorem{corollary}[theorem]{Corollary}
\theoremstyle{remark}

\newtheorem{assumption}[theorem]{Assumption}

\usepackage{amsmath, amsthm, comment}
\usepackage{chngcntr}
\usepackage{apptools}
\usepackage{multirow}

\usepackage[colorinlistoftodos]{todonotes}

\allowdisplaybreaks
\RequirePackage[colorlinks,citecolor=blue,urlcolor=blue, linkcolor = blue ]{hyperref}

\newcommand{\round}[1]{\ensuremath{\lfloor#1\rceil}}

\begin{document}

\def\spacingset#1{\renewcommand{\baselinestretch}%
{#1}\small\normalsize} \spacingset{1}


\if1\blind
{
  \title{\bf Policy effect evaluation under counterfactual neighborhood interventions in the presence of spillover}
 \author{Youjin Lee\thanks{Department of Biostatistics, Brown University. Email: youjin\_lee@brown.edu},~Gary Hettinger\thanks{Department of Biostatistics, Epidemiology and Informatics, University of Pennsylvania. Email: ghetting@pennmedicine.upenn.edu }, and Nandita Mitra\thanks{Department of Biostatistics, Epidemiology and Informatics, University of Pennsylvania. Email: nanditam@pennmedicine.upenn.edu} }
    \date{}
  \maketitle
} \fi

\if0\blind
{
  \bigskip
  \bigskip
  \bigskip
  \begin{center}
    {\LARGE\bf Policy effect evaluation under counterfactual neighborhood interventions in the presence of spillover}
\end{center}
  \medskip
} \fi

\bigskip
\begin{abstract}
Policy interventions can spill over to  units of a population that are not directly exposed to the policy but are geographically close to the units receiving the intervention. In recent work, investigations of spillover effects on neighboring regions have focused on estimating the average treatment effect of a particular policy in an observed setting. Our research question broadens this scope by asking what policy consequences would the treated units have experienced under hypothetical exposure settings. When we only observe treated unit(s) surrounded by controls -- as is common when a policy intervention is implemented in a single city or state -- this effect inquires about the policy effects under a counterfactual neighborhood policy status that we do not, in actuality, observe. In this work, we extend difference-in-differences (DiD) approaches to spillover settings and develop identification conditions required to evaluate policy effects in counterfactual treatment scenarios. These causal quantities are policy-relevant for designing effective policies for populations subject to various neighborhood statuses. We develop doubly robust estimators and use extensive numerical experiments to examine their performance under heterogeneous spillover effects. We apply our proposed method to investigate the effect of the Philadelphia beverage tax on unit sales.  
\end{abstract}

\noindent%
\vfill

\newpage
\spacingset{1.5} 

\section{Introduction}
\label{sec:introduction}

Policy interventions often influence portions of the population who are not directly targeted by the policy. For example, when recreational marijuana is legalized in one state, residents in neighboring states may be willing to cross state borders to buy products that are illegal in their own states, which increases illegal marijuana consumption in neighboring states~\citep{wu2020spillover}. As another example, California’s shelter-in-place orders could also affect health and behavioral outcomes in its neighboring states~\citep{berry2021evaluating}.  These examples illustrate how the effect of a policy intervention can ``spill over'' to populations not directly targeted by the policy but are nonetheless close to the units most directly affected (or treated). Such effects are referred to as \textit{spillover effects}~\citep{tchetgen2012causal, bowers2013reasoning, aronow2017estimating}.
Despite the possibility of spillover effects, the populations located near the directly-affected units have frequently been used as controls because their demographic characteristics are likely to be comparable to those of the treated ones (e.g.,~\cite{campbell2017sharing} and~\cite{thorpe2020evaluation}). 
However, when the policy effects do spill over, these neighbors may fail to provide a comparable control that accurately reflects the situation \textit{in the absence of} the treatment effect; thus, treating them as controls may result in bias when evaluating the effect of a policy intervention. In recent years, focus has turned to examining the effect of a policy on such neighboring controls (rather than on the directly-affected or -treated population) to better understanding the overall impact of a policy intervention~\citep{verbitsky2012causal, berg2020handling, butts2021difference, hu2022average, hettinger2023estimation}.

The Philadelphia beverage tax study motivated our current work~\citep{roberto2019association, bleich2020association, lawman2020one, gibson2021no}. 
\cite{roberto2019association} studied changes in beverage prices (pass-through to the consumer), volume sales, and unit sales following roll out of the sugar-sweetened and artificially-sweetened beverage tax in Philadelphia in January 2017. The study showed that the implementation of an excise tax on sugar-sweetened and artificially-sweetened beverages in the city of Philadelphia was associated with a substantial decrease in unit sales for taxed beverages in Philadelphia but also with an increase in unit sales in bordering counties of Philadelphia that were not subject to the excise tax. The latter association, termed \textit{bypass effect} by~\cite{hettinger2023estimation}, could potentially be explained by cross-border shopping behaviors among Philadelphia city residents who travel to bordering counties in order to bypass the excise tax, partially basing their decision on factors such as the expected price of the product or transportation costs. 
Figure~\ref{fig:real} illustrates the total unit sales of taxed individual- and family-sized beverages of supermarkets in Philadelphia, Baltimore (the control city), and neighboring counties, respectively, before and after the tax implementation, measured at each 4-week period starting from January 2016. Figure~\ref{fig:real} shows that the unit sales of taxed beverages declined in Philadelphia compared with those in Baltimore; whereas the unit sales increased in neighboring counties of Philadelphia, especially for taxed family-sized  beverages.
\begin{figure}[ht]
\centering
\begin{subfigure}[b]{0.48\textwidth}
\includegraphics[width=\textwidth]{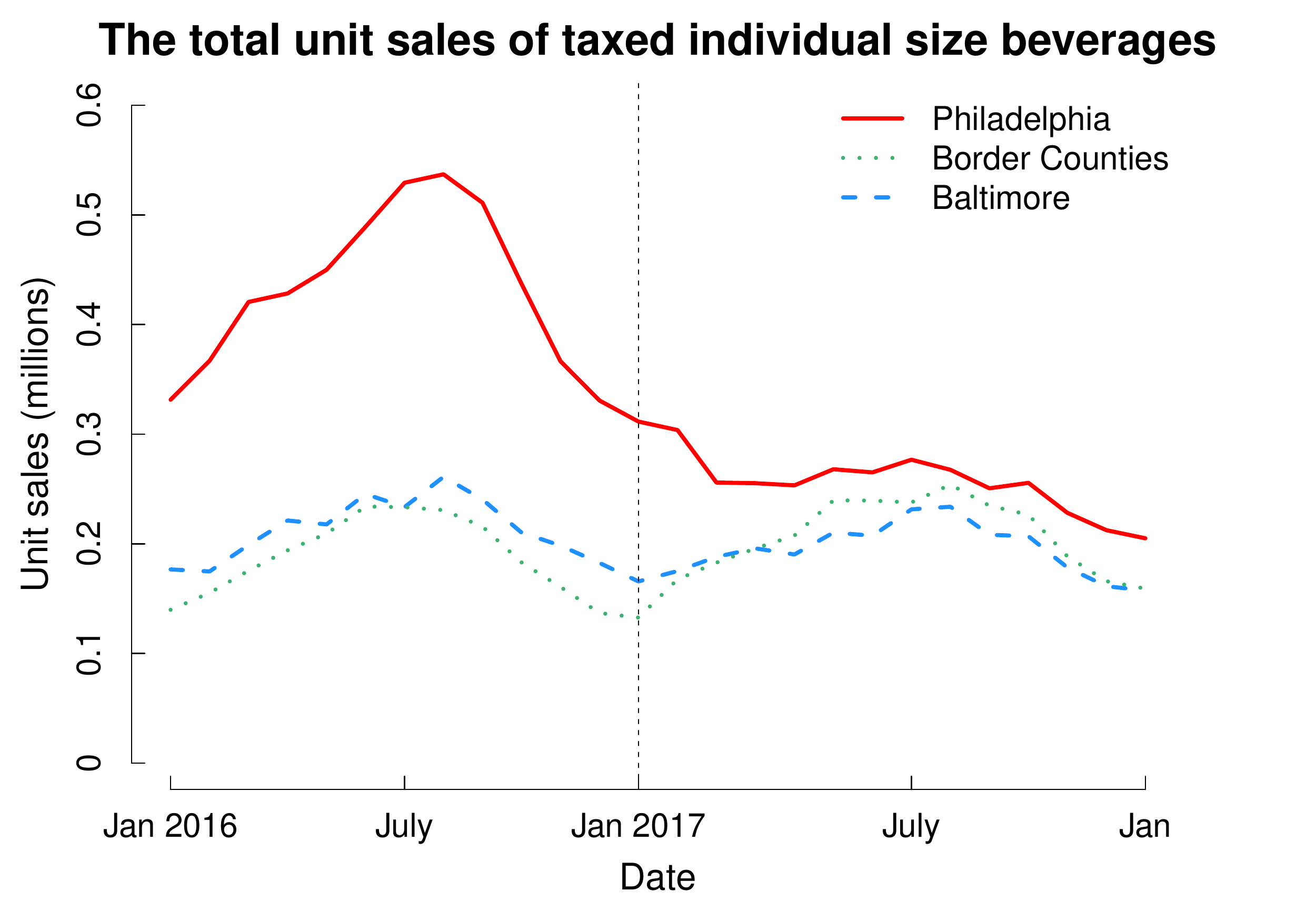}
\end{subfigure}
\begin{subfigure}[b]{0.48\textwidth}
\includegraphics[width=\textwidth]{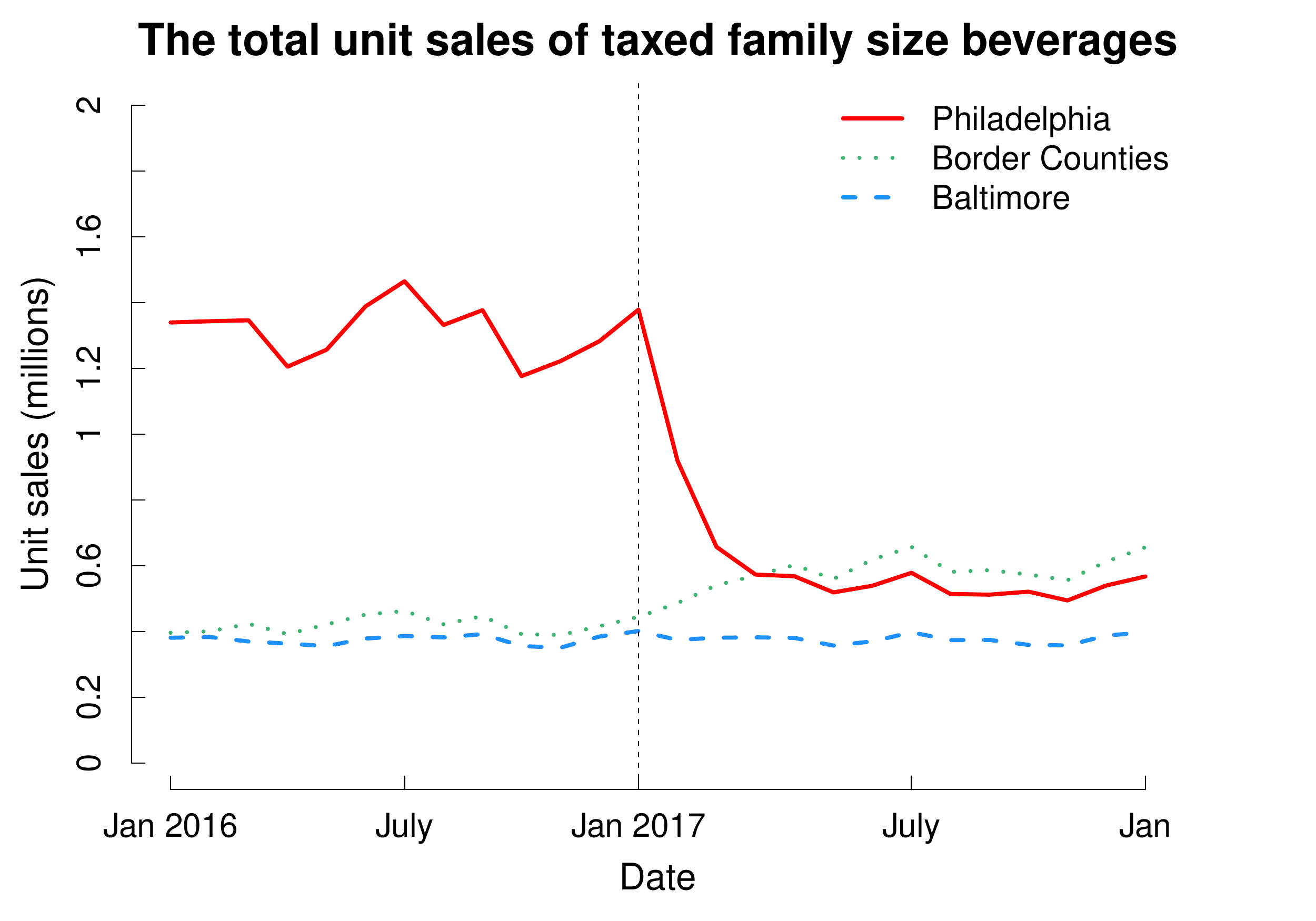}
\end{subfigure}
\caption{\label{fig:real} Changes in the total unit sales of (1) taxed individual sized beverages and (2) taxed family sized beverages in Philadelphia (solid lines), Baltimore (dashed lines), and border counties (dotted lines) supermarkets before and after the excise tax implementation}
\end{figure}

Our follow-up question, then, is as follows: what policy consequences would Philadelphia experience if its surrounding neighbors were also directly exposed to the policy? Answering this question is tantamount to projecting policy effects under a \textit{counterfactual} neighborhood policy status that we do not, in actuality, observe. Because we do not observe the policy status of any bordering counties that were directly affected by the intervention (as there are none), additional identification assumptions would be required in order to identify the potential effect that the policy would have on Philadelphia if it were surrounded by directly-affected neighbors. 
These types of effect evaluations are policy-relevant for two reasons: first, in determining the impact of a policy when a larger population (e.g., not just one city) is affected; and second, in developing effective policies for populations subject to various neighborhood statuses in the future. For example, based on the findings from the Philadelphia beverage tax, researchers or policymakers may want to decide whether or not to expand the beverage tax to other cities or states where neighboring regions have already been affected by an excise tax. Ignoring potential discrepancies in cross-border behaviors between Philadelphia and those cities or states due to different neighborhood tax statuses may result in ineffective policy implementation.

Estimating spillover effects of a policy intervention has only recently been investigated despite the fact that researchers in social sciences and public health have long recognized that a policy change in one region can influence its neighbors~\citep{verbitsky2012causal, delgado2015difference, clarke2017estimating, berg2020handling, butts2021difference}.
The difference-in-differences (DiD) approach is one of the most popular methods to evaluate the causal effects of policy interventions~\citep{abadie2005semiparametric, athey2006identification, donald2007inference}, and the approach has been adapted to measure spillover effects recently.
\cite{clarke2017estimating} discussed a spillover-robust DiD method with a focus on determining an ``optimal'' distance of neighbors indirectly affected by the policy. 
\cite{hettinger2023estimation} developed a doubly-robust DiD estimator for spillover effects on neighboring controls who are adjacent to the treated units. 
\cite{butts2021difference} introduced several causal estimands in the contexts of spatial spillover where the treatment effects spill over from one region to another. However, the causal estimands proposed in \cite{butts2021difference} are based on the \textit{observed} intervention status of neighboring units, not on the \textit{counterfactual} intervention status.

In this work, we propose a new causal estimator that can identify the policy effect on treated units under a counterfactual neighborhood treatment status. We assume that indirectly-affected controls and unaffected controls are fixed and known. We adapt DiD approaches to spillover settings and introduce identification conditions needed to evaluate policy effects under counterfactual treatment scenarios. Using our methodology, researchers will be able to tailor their analyses based on their own specified assumptions about spillover behaviors (e.g., cross-border shopping behaviors) and obtain estimates of policy effects based on those assumptions. These causal quantities can then be utilized to design  effective public policies under diverse neighborhood treatment contexts.

In Section~\ref{sec:setting}, we introduce notation and assumptions followed by our policy-relevant target estimands and their identification conditions. 
In Section~\ref{sec:method}, we propose our new DiD estimator and discuss its properties. In Section~\ref{sec:simulation}, we examine the performance of our proposed doubly-robust estimator through simulation studies. We apply our proposed method to the Philadelphia beverage tax data in Section~\ref{sec:data}. Finally, in Section~\ref{sec:discussion}, we discuss  potential limitations of our proposed method, and suggest future directions for policy effect evaluation in the presence of spillover.

\section{Setting and Target Estimands}
\label{sec:setting}

Throughout this work, we will consider the unit of analysis to be a city, state, or store (i.e., aggregate data), instead of human subjects (e.g., consumers). This is because human subjects are easily movable from the policy exposure (treatment group) to the control group, and vice versa, possibly preventing the treatment level from being well-defined.

Let $Y_{it}$ denote the outcome variable (e.g., unit sales of the taxed product) at time $t = 0,1,2,\ldots, T$ of unit $i=1,2,\ldots,n$. For simplicity, consider two time points, with $t=0$ indicating pre-treatment and $t=1$ indicating post-treatment period. 
Let $Z_{it}$ denote a time-varing binary treatment indicator so that $Z_{it}=1$ if a unit $i$ is treated at time $t$. Let $\bm{Z}_{t} = (Z_{1t}, Z_{2t}, \ldots, Z_{nt})$ be a treatment status vector of $n$ units at time $t$. 
We also have a group indicator where $A_{i} = 1$ indicates being in the treatment group and $\bm{A} = (A_{1}, A_{2}, \ldots, A_{n})$ is a group indicator vector of $n$ units. Then the treatment status is $Z_{it} = A_{i} \times I(t > 0)$. We also observe a $q$-dimensional vector of pre-treatment covariates $\bm{X}_{i} \in \mathbb{R}^{q}$. We will allow time-varying effects of $\bm{X}_{i}$ on the outcome $Y_{it}$ (e.g., time-varying effects of baseline prices  on later sales). One of the baseline covariates that may be relevant is the shortest distance from unit $i$ to the treatment group (e.g., distance to the border), say $X_{ij}$. Then $X_{ij} = 0$ for unit $i$ if $A_{i} = 1$.
Let $\Omega_{C}$ be a sample space of a set $C$.
In summary, we have independent and identically distributed (i.i.d.) panel data $\{ Y_{i0}, Y_{i1}, A_{i}, \bm{X}_{i}  \}_{i=1}^{n}$.

We introduce a potential outcomes framework~\citep{rubin1974estimating, holland1986statistics} to define our target estimand and identification conditions. Under spillover settings, the Stable Unit Treatment Value Assumption (SUTVA) is not satisfied as other units' policy intervention status would have an impact on one's (potential) outcomes.  
Here we consider the standard assumption that the potential outcome at $t$ is only affected by the current (at $t$) treatment assignment. 
Therefore, we can represent each unit's potential outcomes through the entire vector of treatment assignments, $\bm{Z}_{t}$, not only through unit $i$'s own treatment status, $Z_{it}$.
As a result, each unit $i$ has a potential outcome $Y^{(\bm{z}_{t})}_{it}$ at time $t$ if the entire population's treatment status (e.g., tax status) is given by $\bm{z}_{t} \in \{0, 1\}^{n}$, which results in $2^{n}$ potential outcomes for each unit. 
To simplify, we assume that in addition to $Z_{it}$, the entire treatment status $\bm{Z}_{t}$ can only affect the outcome $Y_{it}$ through a \textit{known} function  $h_{it}: \{ 0, 1\}^{n} \rightarrow \mathbb{R}$. This function maps the entire treatment status $\bm{Z}_{t}$ to a scalar-value that denotes the exposure level received by unit $i$ that is not directly through $Z_{it}$. This function is referred to as the \textit{exposure mapping}~\citep{aronow2017estimating}. Without spillover effects, $h_{it}(\bm{z}) = 0$ for all $\bm{z} \in \Omega_{\bm{Z}}$ as there is no exposure received by unit $i$ not directly through its own exposure. 
More examples of $h_{it}$ functions considered in spatial spillover settings can be found in~\cite{butts2021difference}.
Determining or estimating the exposure mapping is dependent on the specific context of the policy evaluation. In our specific context, we assume that $h_{it}(\bm{Z}_{t}) = 0$ for $t=0$, and when $t=1$ (i.e., post-treatment period):

\begin{eqnarray}
\label{eq:exposure}
h_{it}(\bm{Z}_{t}) &=& \left\{ \begin{array}{lll} 1 & \quad \text{Adjacent to the treated unit(s) and } ~Z_{it} = 0  \\ 0 & \quad \text{Not adjacent to the treated unit(s) and }~Z_{it} = 0  \\ 1 & \quad \text{Surrounded by the treated unit(s) and } ~Z_{it} = 1 \\ 0 & \quad \text{Otherwise and} ~Z_{it} = 1     \end{array} \right.
\end{eqnarray}

According to this specification, the neighborhood treatment status is active (i.e., $h_{it}(\bm{Z}_{t}) = 1$) for controls if and only if they are adjacent to at least one treated unit. 
On the other hand, the neighborhood treatment status is active for the treated if and only if they are completely surrounded by treated neighbors. When the treated unit is adjacent to at least one control neighbor, we assume that neighborhood treatment status is inactive (i.e., $h_{it}(\bm{Z}_{t}) = 0$). This definition is largely motivated by the Philadelphia beverage tax study, where having at least one neighbor that is not affected by the tax can be an opportunity for cross-border shopping for Philadelphia residents. Therefore, the treated units that are not completely surrounded by the treated perceive the neighborhood status as inactive, whereas the controls that are adjacent to at least one treated unit perceive the neighborhood status as active. Let $g_{it}(\bm{Z}_{t}) = (Z_{it}, h_{it}(\bm{Z}_{t})) \in \mathbb{R}^{2}$.
Similarly, define time-invariant functions $h_{i}(\bm{A})$ and $g_{i}(\bm{A})$: $h_{i}(\bm{A}) = 1$ if $A_{i}=0$ and unit $i$ is adjacent to treated units or if $A_{i} = 1$ and unit $i$ is surrounded by treated units.  
Note that $A_{i}= 1$ implies $g_{i}(\bm{A}) = (1,0)$ in our case as we do not observe treated units surrounded by other treated units. We extend binary neighborhood treatment assignments to a continuous treatment setting (for example, allowing $p\%$ of the neighbors surrounding unit $i$ to be treated) in our Supplementary Materials (Section~\ref{sec:continuous}). 

One causal estimand we can consider is the average treatment effect on the treated (ATT):
\begin{eqnarray}
    \label{eq:ATT}
    \text{ATT} &:=& E( Y^{(1,0)}_{1} - Y^{(0,0)}_{1}  \mid A = 1 ).
\end{eqnarray}
All expectations throughout this work are taken across units $i=1,2,\ldots,n$, and for simplicity, we omit the unit subscript $i$ within the expectation.
One may be also interested in the spillover effects of a policy intervention on neighboring controls, which we termed the average treatment effect on a neighboring Control (ATN) in previous work ~\citep{hettinger2023estimation}:
\begin{eqnarray}
\label{eq:ATN}
    \text{ATN} = E( Y^{(0,1)}_{1} - Y^{(0,0)}_{1}  \mid g(\bm{A}) = (0,1) ).
\end{eqnarray}
Another version was referred to as the average treatment effect on the close to treated (ATC) in~\cite{clarke2017estimating}. \cite{butts2021difference} generalized this spillover effect as a function of treatment status $Z_{it}$, i.e., $\text{ATN}(z) = E(Y^{(z,1)}_{1} - Y^{(0,0)}_{1}  \mid g(A) = (z,1) )$ given that units with $g(A) = (z,1)$ are available. 

These existing works have attempted to measure the ATT and ATN by adapting DiD approaches. \cite{clarke2017estimating} estimated these effects using the popular two-way fixed effects model after determining the degree of spillover (i.e., how far intervention effects would spill over from the treatment group to its neighboring controls).
\cite{butts2021difference} also used two-way fixed effects models, which allowed for interactions between $A_{i}$ and $h_{i}(\bm{A})$ and incorporated correlation due to repeated measurements and spatial dependence. 
\cite{hettinger2023estimation} improved on these outcome model-dependent approaches by including the treatment assignment model in their estimator, leading to doubly-robust properties.

\subsection{Target estimand}
\label{ssec:AOTT}

In this work, instead of the ATT and ATN, we focus on a new causal estimand that measures the treatment effect when the treated unit perceives that its neighboring units are all treated while keeping our target population as the treated unit(s) with $A_{i} = 1$ (e.g., Philadelphia).
We call this effect the average \textit{offset} treatment effect on the treated (AOTT), denoted by $\tau$.
\begin{eqnarray}
\label{eq:AOTT}
\tau &:=& E ( Y^{(1,1)}_{1} - Y^{(0,0)}_{1} \mid A = 1 ) \nonumber \\ 
&=& \underbrace{E ( Y^{(1,0)}_{1} - Y^{(0,0)}_{1} \mid A = 1 )}_{\text{ATT}} - \underbrace{E ( Y^{(1,0)}_{1} - Y^{(1,1)}_{1} \mid A = 1 )}_{\text{Offsetting effect}}.
\end{eqnarray}
This effect can be decomposed by taking the difference of the ATT and $E(Y^{(1,0)}_{1} - Y^{(1,1)}_{1} \mid A = 1)$. We call the latter expectation the \textit{offsetting effect} of having neighboring controls compared to having neighbors all treated. The opposite sign of this effect, i.e., $E(Y^{(1,1)}_{1} - Y^{(1,0)}_{1} \mid A = 1 )$, was referred to as ``spillover effects on the treated'' in~\cite{butts2021difference}. 
The offsetting effect would have been ``returned'' to the treatment group if the treatment group had been surrounded by treated units rather than being adjacent to controls. Therefore, by subtracting the offsetting effect, the AOTT prevents us from overestimating the treatment effect on the treated.
This estimand is conceptually similar to the ``unrealized spillover causal effect on the treated'' that was defined in~\cite{grossi2020synthetic}. The unrealized spillover causal effect measures the spillover effect on a single treated unit under the hypothetical scenario where a particular unit in the same cluster was exposed to the intervention under a partial interference assumption~\citep{sobel2006randomized}. 

The following three assumptions state that the entire treatment assignment vector affects the potential outcome only through $Z_{it}$ and $h_{it}(\bm{Z}_{t})$ (i.e., $g_{it}(\bm{Z}_{t})$), and this potential outcome is well-defined and observed when the treatment assignment coincides with $g_{it}(\bm{Z}_{t})$.

\begin{assumption}[Neighborhood exposure mapping] \label{as:mapping}
The data $\{Y_{i1}, Y_{i0}, g_{i}(\bm{A}), \bm{X}_{i}\}_{i=1}^{n}$ are i.i.d. with $g_{i}(\bm{A}) = (A_{i}, h_{i}(\bm{A}))$ and
\[
Y^{(\bm{z}_{t})}_{it} = Y^{(z_{it}, h_{it}(\bm{z}_{t}))}_{it} \text{ for all } \bm{z}_{t} \in \Omega_{\bm{Z}_{t}}.
\]
\end{assumption}

\begin{assumption}[Consistency] \label{as:consistency}
\[
Y_{it} = Y^{(z_{it}, h_{it})}_{it} I(Z_{it} = z_{it}, h_{it}(\bm{Z}_{t}) = h_{it}).
\]
\end{assumption}

In addition, we require a positivity assumption on the conditional probability of each exposure given observed covariates.

\begin{assumption}[Positivity] \label{as:pos} For some $\epsilon > 0$,
\[
Pr(A_{i}  = 1) > \epsilon \text{ and } Pr( g_{i}(\bm{A}) = (a_{1}, a_{2})  \mid \bm{X}_{i} ) > \epsilon \text{ for } (a_{1}, a_{2}) = (0, 1), (0,0). 
\]
\end{assumption}

The counterfactual estimand in~\eqref{eq:AOTT} may seem analogous to the ``total effect'' considered in~\cite{butts2021difference}, which is based on the treatment group comprised both of units $i$ with $g_{i}(\bm{A}) = (0,1)$ and $(1,1)$. However, because we do not observe the outcome under $g_{i}(\bm{A}) = (1,1)$ in many situations, such as the Philadelphia beverage tax study, these two estimands are conceptually different. The following additional assumptions are required to identify $E(Y^{(1,1)}_{1} \mid A = 1)$ in~\eqref{eq:AOTT}.

\begin{assumption}[Conditional counterfactual parallel trends with the treated] 
\label{as:pt}
    \[
    E( Y^{(0,0)}_{1} - Y^{(0,0)}_{0} \mid A =1, \bm{X} )  = E(  Y^{(0,0)}_{1} - Y^{(0,0)}_{0} \mid g (\bm{A}) = (0,0), \bm{X} ).
    \]
\end{assumption}
Assumption~\ref{as:pt} allows us to identify the ATT in Equation~\eqref{eq:AOTT}, which allows us to incorporate  time-varying effects of baseline covariates on the outcome.

\begin{assumption}[Counterfactual offsetting effect]    
\label{as:link}
\[
   E(Y^{(1,0)}_{1} - Y^{(1,1)}_{1} \mid A = 1, \bm{X}) +   E(Y^{(0,1)}_{1} - Y^{(0,0)}_{1} \mid g(\bm{A}) = (0,1), \bm{X}) = 0.
\]
\end{assumption}

Assumption~\ref{as:link} allows us to identify the conditional offsetting effects through conditional spillover effects. 
If there are no spillover effects, we have no offsetting effects that would offset the ATT, i.e., $E(Y^{(1,0)}_{1} - Y^{(1,1)}_{1} \mid A = 1, \bm{X} = \bm{x}) = 0$ for all $\bm{x} \in \Omega_{\bm{X}}$. 
This is reasonable as there would be no impact of the neighborhood's policy intervention status in the absence of cross-border shopping behaviors.
On the other hand, if there is no effect heterogeneity in spillover effects, i.e., when $E(Y^{(0,1)}_{1} - Y^{(0,0)}_{1} \mid g(\bm{A}) = (0, 1), \bm{X} = \bm{x}) = E(Y^{(0,1)}_{1} - Y^{(0,0)}_{1} \mid g(\bm{A}) = (0, 1))$ for all $\bm{x} \in \Omega_{\bm{X}}$, then the sum of the offsetting effect and the ATN would be zero. That is, if the treatment group had been surrounded by other directly exposed units, all the spillover effects seen in the nearby control group would have been returned to the treatment group. For example, in a scenario where cross-border shopping behaviors do not depend on baseline characteristics (such as cost of transportation and/or the baseline price of the taxed products), the \textit{entire} change in sales owing to those behaviors would not have happened if the neighboring counties were also directly affected by the policy, i.e., $E(Y^{(1,1)}_{1} - Y^{(0,0)}_{1} \mid A = 1)$ = ATT + ATN. 

In the presence of spillover effect heterogeneity, we assume that the spillover effects observed in neighboring control units with baseline characteristics $\bm{X}$ would be the same as the effect of treating neighboring controls. In other words, there would be a (negative) offsetting effect on treated units with similar baseline characteristics.
Therefore, the sum of the marginalized spillover effects (the ATN) and the marginalized offsetting effect is not necessarily zero.
To identify the ``unrealized spillover effects on the treated'',~\cite{grossi2020synthetic} imputed the potential outcome of the treated unit using a linear combination of the estimated regression coefficients from the model that relates the counterfactual outcome with the pre-treatment outcomes and covariates, distance between units, and duration of treatment. Here, we essentially impute the average potential outcomes of the treated unit under the counterfactual scenario using the spillover effects of neighbors with comparable baseline characteristics $\bm{X}$. Therefore, Assumption~\ref{as:link} can be justified when the baseline covariates $\bm{X}$ are carefully chosen under the assumption that similarity in the observed covariates $\bm{X}$ between the treated and neighboring control groups makes the (negative) offsetting and spillover effects comparable. We provide a detailed discussion on Assumption~\ref{as:link} along with toy examples in Section~\ref{ssec:assumption2.5}.


Next we consider another conditional counterfactual parallel trends assumption, now between the neighboring controls and non-neighboring controls. 
\begin{assumption}[Conditional counterfactual parallel trends with the neighboring controls] 
\label{as:pt2}

Suppose that the following counterfactual parallel trend assumption holds for all $\bm{X} \in \Omega_{\bm{X}}$.
\begin{eqnarray*}
    E( Y^{(0,0)}_{1} - Y^{(0,0)}_{0} \mid  g(\bm{A}) = (0,1),\bm{X} )  = E(  Y^{(0,0)}_{1} - Y^{(0,0)}_{0} \mid g(\bm{A}) = (0,0), \bm{X}).
\end{eqnarray*}
\end{assumption}
Then under the above conditions, we have the following equations for the AOTT conditional on $\bm{X}$:
\begin{eqnarray*} 
&&  E (Y^{(1,1)}_{1} - Y^{(0,0)}_{1} \mid A = 1, \bm{X} ) \\
	&=& E( Y^{(1,0)}_{1} - Y^{(0,0)}_{0} \mid A =1, \bm{X} ) + E (Y^{(1,1)}_{1} - Y^{(1,0)}_{0} \mid A = 1, \bm{X} )\\
	&=& \underbrace{E( Y^{(1,0)}_{1} - Y^{(0,0)}_{0} \mid A = 1, \bm{X} )- E(  Y^{(0,0)}_{1} - Y^{(0,0)}_{0} \mid g(\bm{A}) = (0,0), \bm{X} )}_{\text{ \textbf{ATT}:  by (A4) }} \\ 
	&+ & \underbrace{E( Y^{(0,1)}_{1} - Y^{(0,0)}_{0} \mid g(\bm{A})=(0,1), \bm{X} )- E(  Y^{(0,0)}_{1} - Y^{(0,0)}_{0} \mid g(\bm{A}) = (0,0), \bm{X} )}_{ \textbf{-Offsetting effect}:   \text{ by (A5) and (A6)}}
\end{eqnarray*}

\subsection{Counterfactual offsetting effect}
\label{ssec:assumption2.5}

In this section, we describe Assumption~\ref{as:link} in further detail and illustrate the difference between spillover effects and (negative) offsetting effects using a toy example. 
There are two factors that differentiate these two effects: (1) different distributions of $\bm{X}$ among treatment and neighboring control groups and (2) heterogeneity in spillover effects across $\bm{X}$.

\begin{table}[H]
	\centering
	\resizebox{1.0\textwidth}{!}{\begin{tabular}{c|c|c|c}
			$\bm{x}$ & $Pr(\bm{X} = \bm{x} \mid A = 1)$ & $Pr(\bm{X} = \bm{x} \mid g(\bm{A}) = (0,1))$ & $E(Y^{(0,1)}_{1} - Y^{(0,0)}_{1} | g (\bm{A}) = (0,1), \bm{X} = \bm{x})$ \\ 
			\hline
			0.1 & 0.00 & 0.25 & 0.55 \\
			0.2 & 0.05 & 0.20 & 0.50 \\ 
			0.3 &  0.20 & 0.15 & 0.45\\
			0.4 &  0.30 & 0.10 & 0.30\\ 
			0.5 & 0.25 & 0.10 & 0.20\\ 
			0.6 & 0.20 & 0.10 & 0.10\\
			0.7 & 0.00 & 0.10 & 0.07 \\
			\hline
		 Total & 1.00 & 1.00 &   \\
		 \hline
	\end{tabular}}
	\caption{\label{tab:example} A toy example showing heterogeneity in baseline covariate distributions and spillover effects}
\end{table}

Table~\ref{tab:example} illustrates a hypothetical case with a univariate $X$ (e.g., baseline prices of taxed products), its conditional distributions in the treatment group ($A = 1$) and neighboring control group ($g(\bm{A}) = (0,1)$), and the conditional spillover effects in the neighboring control. 
As the distribution of $\bm{X}$ varies across these two groups, the second and third columns are different. For example, no unit in the treatment group has an $X$ of less than 0.2, whereas a quarter of the units in the neighboring control have an $X$ of less than 0.2. 
Assumption~\ref{as:link} indicates that the opposite sign of the offsetting effect (i.e., $E\{E(Y^{(1,1)}_{1} - Y^{(1,0)}_{1} \mid A = 1, \bm{X})\}$) marginalizes the conditional spillover effects (the fourth column) over the treatment group (i.e., $E\{ E(Y^{(0,1)}_{1} - Y^{(0,0)}_{1} \mid g(\bm{A}) = (0, 1), \bm{X} ) \mid A = 1 \}$), whereas the ATN (Equation~\eqref{eq:ATN}) marginalizes the conditional spillover effects over the neighboring control group. 
In this example, the marginalized offsetting effect, -$E \{ E(Y^{(0,1)}_{1} - Y^{(0,0)}_{1} | g(\bm{A}) = (0,1), \bm{X}) \mid A = 1 \}$, is given by -0.275. Observe that only the conditional distribution of $\bm{X}$ on the treatment group, and not on the neighboring control group, is relevant here.
On the other hand, the ATN (i.e., the marginalized spillover effect on neighboring controls) is given by $E \{ E(Y^{(0,1)}_{1} - Y^{(0,0)}_{1} | g(\bm{A}) = (0,1), \bm{X}) \mid g(\bm{A}) = (0,1) \} = 0.372$. Therefore, the summation of the ATN and offsetting effect is not necessarily zero even under Assumption~\ref{as:link}. In the absence of spillover effect heterogeneity (i.e., if the fourth column in Table~\ref{tab:example} were the same for all $x$ values), the sum of these two effects should be zero, i.e., all the spillover effects observed in neighboring controls would have been returned to Philadelphia under the counterfactual scenario. 

Depending on an investigator's prior knowledge or beliefs, baseline covariates $\bm{X}$ can be properly redefined so as to make Assumption~\ref{as:link} more plausible. One can take into account multiple social and economic factors related to spillover behaviors, such as transportation costs~\citep{knight2013state} and distance to the border~\citep{asplund2007demand}, and then combine these into one variable. In this way, one can define $\bm{X}$ such that similarities in $\bm{X}$ imply comparability between offsetting and spillover effects so as to satisfy Assumption~\ref{as:link}.  
For instance, suppose that the baseline covariates include transportation costs for cross-border shopping behaviors and the baseline price of the product. Then these multivariate covariates could be combined into a univariate covariate, $X$, that indicates the ``total cost'' of buying the product.
For the treatment group, $X$ would then denote the baseline price of the taxed products while for the neighboring controls group it would denote the baseline price of the taxed products plus transportation costs.
Under this covariate definition of $X$, we are assuming that the extra purchases observed in the neighboring control with the total cost of $X$ (price of the taxed products plus transportation costs) would have returned to the treatment group with the baseline price of $X$ (with a zero transportation cost) when all of its neighbors had received the policy intervention, i.e., those conditional spillover effects are comparable to the conditional offsetting effect with the same $\bm{X}$.

\begin{table}[H]
	\centering
	\resizebox{1.0\textwidth}{!}{\begin{tabular}{c|c|c|c||c|c|c|c}
	
			\multicolumn{4}{c||}{$A = 1$} & \multicolumn{4}{c}{$g(\bm{A}) = (0,1)$} \\
			\hline
		Distance  & Price & $X = x$ & $Pr(X = x \mid A = 1)$ & Distance & Price & $X = x$ & $Pr(X = x \mid g(\bm{A}) = (0,1) )$ \\
			\hline
		0.00 & 0.1 & 0.1 & 0.00 & 0.30 & 0.1 & 0.1 & 0.25\\
		0.00 & 0.2 & 0.2 & 0.05 & 0.45 & 0.2 & 0.3 & 0.20 \\ 
		0.00 & 0.3 & 0.3 & 0.20 & 0.50 & 0.3 & 0.4 & 0.15\\
		0.00 & 0.4 & 0.4 & 0.30 & 0.35 & 0.4 & 0.5 & 0.10\\
		0.00 & 0.5 & 0.5 & 0.25 & 0.40 & 0.5 & 0.6 & 0.20 \\ 
		0.00 & 0.6 & 0.6 & 0.20 & 0.20 & 0.6 &  & \\
		0.00 & 0.7 & 0.7 & 0.00 & 0.05 & 0.7 & 0.7 & 0.10\\
			\hline

			\hline
	\end{tabular}}
	\caption{\label{tab:example2} A toy example showing the construction of the baseline covariate $X$. Here distance indicates the distance from the border of the treated unit.}
\end{table}

Table~\ref{tab:example2} illustrates the case where the baseline covariate $X$ incorporates the information on the distance from the border into the price.
Let us redefine $X$ as follows: if the distance from the border is 0.3 or less than 0.3, $X$ simply denotes the price; if not, we add price by 0.1. The additional cost of 0.1 can be compared to the additional travel expenditures that consumers would incur. Then the extra purchases made by consumers who traveled to neighboring controls with a distance of 0.45 and a baseline price of 0.2 as a result of the policy intervention are \textit{comparable} to the loss in purchases in the treatment group with a baseline price of 0.3 as a result of having neighboring controls as opposed to having all treated neighbors in this setting.

There are a few things we need to consider in constructing $\bm{X}$.
The first is the positivity assumption (Assumption~\ref{as:pos}), which necessitates that the propensity scores of units in either the neighboring or non-neighboring control groups are not zero. Essentially, there should be overlap in the support of $\bm{X}$ between the two control groups. 
Sometimes, redefining $\bm{X}$ can make the positivity assumption more plausible. For instance, including distance itself as a variable in $\bm{X}$ would violate the positivity assumption since there are no controls with a distance of zero in Table~\ref{tab:example2}; whereas, incorporating distance information in $\bm{X}$ in the manner described above can still satisfy the positivity assumption.
We also need to consider how the implications of Assumptions~\ref{as:pt} and~\ref{as:pt2}, which condition on $\bm{X}$, differ when redefining $\bm{X}$. 
For example, in the toy example outlined in Table~\ref{tab:example2}, Assumption~\ref{as:pt2} implies that the outcome trend of the neighboring control with (Price=0.2) and (Distance=0.45) under no intervention is parallel to the outcome trend of the treatment group with (Price=0.3) under no intervention.

\section{Estimators}
\label{sec:method}

In this section, we propose new causal estimators for the offsetting effect and AOTT, and discuss their asymptotic properties. All proofs of theoretical results can be found in the Supplementary Material Section~\ref{sec:proofs}.

\subsection{A doubly-robust estimator}

We first introduce the inverse probability weighted (IPW) estimator for the offsetting effect. 
Let $\pi_{(a_{1}, a_{2})}(\bm{X})$ denote the propensity score $Pr(g(\bm{A}) = (a_{1}, a_{2}) \mid \bm{X})$ and $\pi_{(a_{1}, a_{2})}$ denote the marginal probability $Pr(g(\bm{A}) = (a_{1}, a_{2}))$, each for $(a_{1}, a_{2}) \in  \{ (1,0), (0,1), (1,1) \}$.
Let $\{\pi_{(a_{1}, a_{2})}(\bm{X})\}$ denote $\{ \pi_{(a_{1}, a_{2})} (\bm{X}): (a_{1}, a_{2}) \in \{(1,0), (0,1), (0,0)\};~ t = 0,1; ~\bm{X} \in \Omega_{\bm{X}} \}$, i.e., a collection of the propensity scores. Let $E_{n}(\cdot)$ denote empirical averages.

\begin{lemma}
\label{lemma:ipw}

Suppose that Assumptions~\ref{as:mapping}--\ref{as:pos} and~\ref{as:link}--\ref{as:pt2} hold. Then we have the following IPW estimator that provides a consistent estimator for $\delta = E(Y^{(1,1)}_{1} - Y^{(1,0)}_{1} \mid A = 1)$.
\begin{eqnarray*}
\label{eq:ipw}
\hat{\delta}^{\text{ipw}} = \hat{\pi}^{-1}_{(1,0)} E_{n} \left[ \hat{\pi}_{(1,0)}(\bm{X}) \left\{ \frac{I( g(\bm{A}) = (0,1) }{\hat{\pi}_{(0,1)}(\bm{X})} - \frac{I( g(\bm{A}) = (0,0))}{\hat{\pi}_{(0,0)}(\bm{X}) } \right\} (Y_{1} - Y_{0}) \right].
\end{eqnarray*}
\end{lemma}

We can use a consistent plug-in estimator for $\pi_{(1,0)}$ and $\{\pi_{(a_{1}, a_{2})}(\bm{X})\}$ to build a consistent estimator of $\hat{\delta}^{\text{ipw}}$ when $\hat{\pi}_{(1,0)}(\bm{X}) \overset{p}{\rightarrow} \pi_{(1,0)}(\bm{X})$ 
and $\hat{\pi}_{(0,0)}(\bm{X}) \overset{p}{\rightarrow} \pi_{(0,0)}(\bm{X})$ as $n \rightarrow \infty$. 
In our simulation study, we use the sample average to estimate $\pi_{(1,0)}$ and estimated propensity scores from a multinomial logistic regression model for the $\{\pi_{(a_{1}, a_{2})}(\bm{X})\}$ 's.

Next we propose combining the IPW approach with the outcome regression to construct a doubly-robust estimator for the offsetting effect (i.e., $-\delta$).
Let $\mu^{(a_{1},a_{2})}_{t}(\bm{X})$ denote the true outcome model for $E(Y_{t} \mid g(\bm{A}) = (a_{1},a_{2}), \bm{X} )$ for $t=0,1$ and $(a_{1}, a_{2})\in \{(1,0), (0,1), (0,0)\}$. 
Let  
$\{\mu^{(a_{1}, a_{2})}_{t} (\bm{X})\}$ denote $\{ \mu^{(a_{1}, a_{2})}_{t} (\bm{X}):  (a_{1}, a_{2}) \in \{(1,0), (0,1), (0,0)\}; ~t =0,1; ~\bm{X} \in \Omega_{X} \}$, i.e., a collection of the outcome regression models. 
Let $\Delta Y_{i} := Y_{i1} - Y_{i0}$ and $\Delta \mu^{(a_{1}, a_{2})}(\bm{X}) = \mu^{(a_{1}, a_{2})}_{1}(\bm{X}) - \mu^{(a_{1}, a_{2})}_{0}(\bm{X})$.

\begin{theorem}
\label{thm:dr}
Suppose that Assumptions~\ref{as:mapping}--\ref{as:pos} and~\ref{as:link}--\ref{as:pt2} hold. Then $\delta^{\text{dr}}$ is a consistent estimator for $\delta:= E(Y^{(1,1)}_{1} - Y^{(1,0)}_{1} \mid A = 1)$ when either one of the nuisance functions -- $\{\pi_{(a_{1}, a_{2})}(\bm{X})\}$ or $\{\mu^{(a_{1}, a_{2})}_{t}(\bm{X})\}$ -- is correctly specified. That is,  $\delta^{\text{dr}}$ is a doubly-robust estimator for $\delta$. 
\begin{eqnarray*}
\hat{\delta}^{\text{dr}} &=&  E_{n} \left\{ \frac{\hat{\pi}_{(1,0)}(\bm{X} )}{\hat{\pi}_{(1,0)}} \frac{I(g(\bm{A}) = (0,1)) }{\hat{\pi}_{(0,1)}(\bm{X})} \left( \Delta Y - \Delta \hat{\mu}^{(0,1)}(\bm{X}) \right) + \frac{I(A=1)}{\hat{\pi}_{(1,0)}}\Delta \hat{\mu}^{(0,1)}(\bm{X})  \right\} \\ \quad && -
E_{n} \left\{ \frac{\hat{\pi}_{(1,0)}(\bm{X} )}{\hat{\pi}_{(1,0)}} \frac{I(g(\bm{A}) = (0,0)) }{\hat{\pi}_{(0,0)}(\bm{X})} \left( \Delta Y - \Delta \hat{\mu}^{(0,0)}(\bm{X}) \right) + \frac{I(A=1)}{\hat{\pi}_{(1,0)}}\Delta \hat{\mu}^{(0,0)}(\bm{X})  \right\}.
\end{eqnarray*}
\end{theorem}

Similarly, we can obtain $\hat{\delta}^{\text{dr}}$ by plugging in the predicted values from the propensity scores and outcome regression models, i.e., plugging in $\{ \hat{\pi}_{(a_{1}, a_{2})}(\bm{X}) \}$ and $\{ \hat{\mu}^{(a_{1}, a_{2})}_{t} (\bm{X})\}$.


\begin{theorem}
\label{thm:variance}
Assume the same conditions as in Theorem~\ref{thm:dr}. 
Then the efficient influence function for $\delta:= E(Y^{(1,1)}_{1} - Y^{(1,0)}_{1} \mid A = 1)$ is given by: 
\begin{eqnarray*}
&& \psi_{\delta, i}(Y_{1}, Y_{0}, \bm{A}, \bm{X})\\ \quad  &=& \frac{\pi_{(1,0)}(\bm{X}_{i} )}{\pi_{(1,0)}} \frac{I(g_{i}(\bm{A}) = (0,1)) }{\pi_{(0,1)}(\bm{X}_{i})} \left( \Delta Y_{i} - \Delta \mu^{(0,1)}(\bm{X}_{i}) \right) + \frac{I(A_{i}=1)}{\pi_{(1,0)}}\Delta \mu^{(0,1)}(\bm{X}_{i})  \\  && - \left\{
 \frac{\pi_{(1,0)}(\bm{X}_{i} )}{\pi_{(1,0)}} \frac{I(g_{i}(\bm{A}) = (0,0)) }{\pi_{(0,0)}(\bm{X}_{i})} \left( \Delta Y_{i} - \Delta \mu^{(0,0)}(\bm{X}_{i}) \right) + \frac{I(A_{i}=1)}{\pi_{(1,0)}}\Delta \mu^{(0,0)}(\bm{X}_{i})\right\} \\ && - \frac{I(A_{i}=1)}{\pi_{(1,0)}}\delta.  
\end{eqnarray*}

Therefore, the semiparametric efficiency bound for all regular estimators for $\delta$ is given by: 
\begin{eqnarray}
\label{eq:bound}
\sigma^{2}_{\delta_{\text{sb}}} := E\left[ \{ \psi_{\delta, i}(Y_{1}, Y_{0}, \bm{A}, \bm{X}) \}^{2} \right].
\end{eqnarray}
\end{theorem}

\begin{theorem}
\label{thm:asymptotic}
Suppose Assumptions~\ref{as:mapping}--\ref{as:pos} and~\ref{as:link}--\ref{as:pt2}, and standard regularity conditions hold.
Then when both of the nuisance functions (i.e., both of $\{\pi_{(a_{1}, a_{2})}(\bm{X})\}$ and $\{\mu^{(a_{1}, a_{2})}_{t}(\bm{X})\}$) are correctly specified, the following result holds:
\begin{eqnarray*}
\sqrt{n}(\hat{\delta}^{\text{dr}} - \delta) \longrightarrow N(0, \sigma^{2}_{\delta_{sb}}).
\end{eqnarray*}
\end{theorem}

The above theorems state that $\hat{\delta}^{\text{dr}}$ is doubly-robust, and also locally semiparametrically efficient, i.e., its asymptotic variance achieves the semiparametric efficiency bound when the working models for the nuisance functions are correctly specified.
However, in practice, one still needs to choose a particular estimation procedure for the nuisance functions. Moreover, the results of Theorem~\ref{thm:asymptotic} are not guaranteed when either of the nuisance functions is not correctly specified. For more guidance on the choice of the first-step estimators to further improve the standard doubly-robust DiD estimators' doubly-robust inference, see \cite{sant2020doubly}.

Finally, we combine the doubly-robust estimators for the ATT with $\delta$ to obtain the doubly-robust estimator for the AOTT (=ATT + $\delta$):
\begin{eqnarray*}
\label{eq:delta}
\hat{\tau}^{\text{dr}} &=&  E_{n} \left\{ \left(\frac{I(A = 1)}{\hat{\pi}_{(1,0)}}-\frac{\hat{\pi}_{(1,0)}(\bm{X}) I(g(\bm{A}) = (0,1) ) }{ \hat{\pi}_{(1,0)} \hat{\pi}_{(0,1)}(\bm{X}) } \right) (\Delta Y - \Delta \hat{\mu}^{(0,0)} (\bm{X})) \right\} + \hat{\delta}^{\text{dr}}
\end{eqnarray*}

\begin{corollary}
\label{cor:AOTT}
Suppose that Assumptions~\ref{as:mapping}--\ref{as:pt2} hold. Then as the sample size increases, i.e., as $n \rightarrow \infty$, the proposed estimator $\hat{\tau}^{\text{dr}}$ converges to $\tau$ when either one of the nuisance functions (i.e., either $\{\pi_{(a_{1}, a_{2})}(\bm{X})\}$ or $\{\mu^{(a_{1}, a_{2})}_{t}(\bm{X})\}$) is correctly specified. Moreover, the following result holds when both of the nuisance functions are correctly specified and standard regularity conditions are satisfied. 
\begin{eqnarray*}
\sqrt{n}(\hat{\tau}^{\text{dr}} -  \tau) \longrightarrow N(0, \sigma^{2}_{\tau}),
\end{eqnarray*}
where $\sigma^{2}_{\tau}$ is the semiparametric efficiency bound for all regular estimators for $\tau$.
\end{corollary}

\subsection{Extension to multiple time points}
\label{ssec:multiple}

We can extend our proposed estimator to multiple pre- and post-intervention time periods setting. Suppose that we have $T$ pre-treatment time periods (i.e., $t=-(T-1), -(T-2), \ldots, 0$) and $T$ post-treatment time periods (i.e., $t=1,2,\ldots, T$). 
Let us modify Assumptions~\ref{as:pt} and~\ref{as:pt2} as follows.
\begin{assumption}[Conditional counterfactual parallel trends]  For $t=1,2,\ldots, T$:
\label{as:multi_pt}
    \begin{eqnarray*}
    E( Y^{(0,0)}_{t} - Y^{(0,0)}_{t-T} \mid A =1, \bm{X} ) &=& E(  Y^{(0,0)}_{t} - Y^{(0,0)}_{t-T} \mid g(\bm{A}) = (0,0), \bm{X} ) \\
     E( Y^{(0,0)}_{t} - Y^{(0,0)}_{t-T} \mid g(\bm{A})=(0,1), \bm{X} )  &=& E(  Y^{(0,0)}_{t} - Y^{(0,0)}_{t-T} \mid g(\bm{A}) = (0,0), \bm{X} ).
    \end{eqnarray*}
\end{assumption}
Then when $\Delta Y_{it, (T)} :=  Y_{it} - Y_{i,t-T}$ and $\Delta \mu^{(a_{1}, a_{2})}_{t,(T)} (\bm{X}_{i}) := \mu^{(a_{1}, a_{2})}_{t} (\bm{X}_{i}) - \mu^{(a_{1}, a_{2})}_{t-T}(\bm{X}_{i})$ for $t>0$, we have a consistent estimator of $E(Y^{(1,1)}_{t} - Y^{(1,0)}_{t} \mid g_{i}(\bm{A}) = (1,0))$ as follows:
\begin{eqnarray*}
\hat{\delta}^{\text{dr}}_{t} &=&  E_{n} \left\{ \frac{\pi_{(1,0)}(\bm{X} )}{\pi_{(1,0)}} \frac{I(g(\bm{A}) = (0,1)) }{\pi_{(0,1)}(\bm{X})} \left( \Delta Y_{t, (T)} - \Delta \mu^{(0,1)}_{t, (T)}(\bm{X}) \right) + \frac{I(A=1)}{\pi_{(1,0)}}\Delta \mu^{(0,1)}_{t, (T)}(\bm{X})  \right\} \\ &-& 
E_{n} \left\{ \frac{\pi_{(1,0)}(\bm{X} )}{\pi_{(1,0)}} \frac{I(g(\bm{A}) = (0,0)) }{\pi_{(0,0)}(\bm{X})} \left( \Delta Y_{t,(T)} - \Delta \mu^{(0,0)}_{t,(T)}(\bm{X}) \right) + \frac{I(A=1)}{\pi_{(1,0)}}\Delta \mu^{(0,0)}_{t, (T)}(\bm{X})  \right\}, 
\end{eqnarray*}
for $t=1,2,\ldots, T$. We can then similarly define a time-specific doubly robust estimator for the AOTT, $E(Y^{(1,1)}_{t} - Y^{(0,0)}_{t} \mid g_{i}(\bm{A}) = (1,0))$ for $t > 0$, using $\Delta Y_{it, (T)}$ and $\Delta \mu^{(a_{1}, a_{2})}_{t, (T)}$:
\begin{eqnarray*}
\hat{\tau}^{\text{dr}}_{t} &=& E_{n} \left\{ \left(\frac{I(A = 1)}{\pi_{(1,0)}}-\frac{\pi_{(1,0)}(\bm{X} ) I(g(\bm{A}) = (0,1) ) }{ \pi_{(1,0)} \pi_{(0,1)}(\bm{X} ) } \right) (\Delta Y_{t,(T)} - \Delta \mu^{(0,0)}_{t, (T)} (\bm{X} )) \right\}
+ \hat{\delta}^{\text{dr}}_{t}.
\end{eqnarray*}
Finally, we take an average of the time-specific estimator and obtain the time averaged AOTT estimator given by:
\begin{eqnarray}
\label{eq:time_AOTT}
\hat{\tau}^{\text{dr, avg}} &=&  \frac{1}{T}\sum\limits_{t =1}^{T} \hat{\tau}^{\text{dr}}_{t}.
\end{eqnarray}

\section{Simulations}
\label{sec:simulation}

Through simulation studies, we examine the asymptotic behavior of the proposed estimators in finite samples when the distributions of baseline covariates $\bm{X}$ differ between groups and $\bm{X}$ induces heterogeneity in spillover effects.
We consider two time points ($t=0$: pre-intervention; $t=1$: post-intervention) and $n$ units. We allow time-varying effects of $\bm{X}_{i}$ on the outcome in the absence of treatment, as well as on the treatment effects. We first generate two baseline covariates, $X_{i1}$ and $X_{i2}$ as:
\begin{eqnarray*}
( X^{*}_{i1}, X^{*}_{i2})^{T} & \quad \overset{i.i.d.}{\sim} \quad & N \left( \begin{pmatrix} 0 \\ 0 \end{pmatrix} , \begin{pmatrix} 1.0 & ~0.3 \\ 0.3 & ~1.0  \end{pmatrix}  \right); \\
X_{ik} &=& \min ( \max(-2.0,  \round{10 X^{*}_{ik}}/10), 2.0) , \quad  k = 1,2;~i=1,2,\ldots, n.
\end{eqnarray*}
By transforming a continuous $X^{*}_{ik}$ into $X_{ik}$,
we let $X_{ik}$ take a finite number of values and have the minimum and maximum of $X_{ik}$ as $-2.0$ and $2.0$, respectively. This transformation is merely for convenience in redefining $\bm{X}$ tailored to Assumption~\ref{as:link}.

We then consider the following data generating model for treatment assignments. For $(a_{1}, a_{2}) \in \{(1,0), (0, 1) \}$:
\begin{eqnarray*}
\ln\left\{ \frac{Pr( g_{i}(\bm{A}) = (a_{1}, a_{2} ))}{Pr( g_{i}(\bm{A}) = (0,0))} \right\} &=&  \beta_{0,(a_{1}, a_{2})} + X_{1i} \beta_{1,(a_{1}, a_{2})} + X_{2i} \beta_{2, (a_{1}, a_{2})}.
\end{eqnarray*}
where $\boldsymbol{\beta}_{(a_{1}, a_{2})} = (\beta_{0, (a_{1}, a_{2})}, \beta_{1, (a_{1}, a_{2})}, \beta_{2, (a_{1}, a_{2})})$ for $(a_{1}, a_{2}) \in \{(0,0), (1,0), (0,1)\}$ and $\boldsymbol{\beta}^{T}_{(0,0)} = (0, 0,0)$, $\boldsymbol{\beta}^{T}_{(1,0)} = (-0.5, 1,0.5)$, $\boldsymbol{\beta}^{T}_{(0,1)} = (0.3, -0.5, -0.5)$.
Next, given $g_{i}(\bm{A})$ and $\bm{X}_{i}  = (X_{i1}, X_{i2})$, we generate the potential outcomes under the control exposure for units with $g_{i}(\mathbf{A}) = (a_{1}, a_{2})$ as:
\begin{eqnarray*}
 Y^{(0,0)}_{it} &=& \gamma_{t} + X_{1i}\lambda_{1t} + X_{2i}\lambda_{2t} + \alpha_{(a_{1},a_{2})} + \epsilon_{it},
\end{eqnarray*}
for $(a_{1}, a_{2
}) \in \{ (0,0), (0,1), (1,1)\}$ and $t=0,1$. Here, $(\alpha_{(0,0)}, \alpha_{(1,0)}, \alpha_{(0,1)}) = (0.0, 0.5, -1.0)$ represent baseline group effects. We set the time-specific effects as $\gamma_{t=0} = 1$, $\gamma_{t=1} = 1.5$ and the time-varying effects of our covariates as 
$\boldsymbol{\lambda}^{T}_{t=0} = (0.5, 0.5)$, $\boldsymbol{\lambda}^{T}_{t=1} = (-0.5, 1.0)$.
Then, we generate random errors, $\epsilon_{it}$, as:
\begin{eqnarray*}
 (\epsilon_{i1},\epsilon_{i2})^{T} & \quad \overset{i.i.d.}{\sim} \quad & N \left( \begin{pmatrix} 0 \\ 0 \end{pmatrix} , \begin{pmatrix} 1.0 & ~0.2 \\ 0.2 & ~1.0  \end{pmatrix}  \right), ~ i=1,2, \ldots, n.
\end{eqnarray*}
The non-zero correlation of 0.2 reflects temporal correlation within units. 
We then generate the potential outcomes for control units with $g_{i}(\bm{A}) = (0, 1)$ under $g_{it}(\bm{Z}_{t}) = (0,1)$ and for treated units with $g_{i}(\bm{A}) = (1,0)$ under $g_{it}(\bm{Z}_{t}) = (1,0)$ as follows: 
\begin{eqnarray*}
 Y^{(0,1)}_{it}  &=& Y^{(0,0)}_{it} +  \theta_{(0,1)} I( t > 0) \{( 1+0.5X_{i1} ) + (1-2X^2_{i2}) \} \\ 
 Y^{(1,0)}_{it}   &=& Y^{(0,0)}_{it} + \theta_{(1,0)} I(t>0) (1-0.5X_{i1}).
\end{eqnarray*}
We set $\theta_{(0,1)} = 0.5$ and $\theta_{(1,0)} = -1.0$ to reflect that the intervention has a negative effect on the treated group outcomes but a positive effect on the neighboring control group outcomes when $\bm{X}_{i} = (0,0)$. However, these two intervention effects may differ according to covariates $X_{i1}$ and $X_{i2}$.   
The true offsetting and AOTT effects are empirically derived from the true data generating models. The true ATT, offsetting effect, and AOTT are -0.708, -0.183, and -0.526, respectively. 

We apply our doubly-robust estimator $\hat{\tau}^{\text{dr}}$ to estimate the AOTT under the following three scenarios: (a) both the outcome regression model and the propensity score model are correctly specified; (b) the propensity score model is correctly specified but the outcome model omits the interaction terms between $X_{ki}$ and $I(t>0)$ ($k=1,2$) and $X^{2}_{i2}$ and $I(g_{i}(\bm{A}) = (0,1)) I(t>0)$; and (c) the outcome regression model is correctly specified but the propensity score model omits the $X_{i1}$ term and includes $\exp(X_{i2})$ instead of $X_{i2}$. We replicate the experiment $1000$ times for each scenario with different sample sizes ($n=500, 1000$ and $2000$ observations), with aggregate results displayed in Table~\ref{tab:results}. Under all scenarios (a)--(c), the bias converges to zero and the coverage rates of 95\% confidence intervals are close to nominal.

\begin{table}[H]
\centering
\resizebox{\textwidth}{!}{\begin{tabular}{l|ccc|ccc|ccc}
Case: & \multicolumn{3}{c|}{(a) All correct} & \multicolumn{3}{c|}{(b) Incorrect outcome regression} &  \multicolumn{3}{c}{(c) Incorrect propensity scores} \\
 & Bias & SE & CR & Bias & SE & CR & Bias & SE & CR \\
 \hline
 $n=500$ & 0.029 & 0.456 & 0.941 & 0.055 & 0.781 & 0.942 & 0.007 & 0.244 & 0.940 \\
 $n=1000$ & 0.010 &0.329 &0.939 & 0.030 & 0.581 & 0.945 & -0.001  & 0.172 & 0.928\\
$n=2000$  & -0.001  &0.235 & 0.951 & 0.027 & 0.434 & 0.944 & 0.009 & 0.121 & 0.938\\
\hline
\end{tabular}}
\caption{\label{tab:results} The performance of the proposed estimator $\hat{\tau}^{\text{dr}}$ when (a) all of the propensity scores and the outcome regression are correctly specified, (b) the outcome regression is misspecified, and (c) the propensity scores are misspecified. Standard errors (SE) and coverage rates (CR) of 95\% confidence intervals are based on the influence function-based variance estimator}
\end{table}

In the Supplementary Material Section~\ref{sec:additional}, we demonstrate that the IPW and outcome regression-based estimators are sensitive to model misspecifications. 
We also present results from simulations when there are multiple time points in Supplementary Material Section~\ref{sec:multiple}.

\section{Data application}
\label{sec:data}

We apply our proposed method to volume sales data obtained from Information Resources Inc (IRI)~\citep{muth2016understanding} that were used for our previous studies of the Philadelphia beverage tax ~\citep{roberto2019association, bleich2020association, lawman2020one, gibson2021no}. Here, we evaluate the effect of the beverage tax  on Philadelphia under the hypothetical scenario that Philadelphia's neighboring counties were also directly exposed to the policy. 
Table~\ref{tab:2by2} describes the definition of four different treatment groups according to the exposure mapping defined in Equation~\eqref{eq:exposure}. We use Baltimore as a non-neighboring control that satisfies the conditional counterfactual parallel trends assumption (Assumptions~\ref{as:pt} and~\ref{as:pt2}).
\begin{table}[H]
\centering
\caption{\label{tab:2by2} A combination of the four treatment groups in Philadelphia beverage tax study}
\resizebox{0.7\textwidth}{!}{\begin{tabular}{c|c|c}
  \hline
  & $A_{i}=0$ & $A_{i}=1$   \\ 
  \hline
$h_{i}(\bm{A})=0$ &  Isolated control & Philadelphia adjacent to controls \\
\hline
$h_{i}(\bm{A})=1$ &  Neighboring control &  Philadelphia surrounded by the treated \\ 
\hline
\end{tabular}}
\end{table}

For each unit (i.e., store), we consider 26 time points representing 4-week periods with end dates of 01/31/2016, 02/28/2016, $\cdots$, 12/31/2017 ($t=-12, -11, 0, 1, \cdots, 13$). Then, $Z_{it} = 1$ for stores in Philadelphia if the end date of the 4-week period is 01/29/2017 or after (i.e., $t>0$). We focus on the effect of the tax implementation on the total unit sales of (1) taxed individual size beverages and (2) taxed family size beverages in supermarkets. We have 13, 15, and 26 supermarkets in Baltimore, neighboring counties of Philadelphia, and Philadelphia, respectively. 

In the propensity score model, we include weighted price per ounce over each 4-week period (with the end date of 01/01/2016) of family-sized sweetened beverages and individual-sized sweetened beverages as baseline covariates $\bm{X}_{i}$.
In addition, we use publicly available zip-code level characteristics, such as the proportion of Black, Hispanic, Asian, and male residents and average income per household. 
In both outcome models for outcomes (1) and (2), we include the fixed effects of treatment group indicators, $I(t > 0)$, and their interactions with store-level random intercepts. We also include covariate effects for the aforementioned zip-code level characteristics and the baseline ($t=-12$) weighted price per ounce of both family- and individual-sized sweetened beverages as well as the interaction between price and treatment effects. The estimates from the propensity score and the outcome models are presented in Supplementary Material Section~\ref{sec:D}.

Table~\ref{tab:bev_result} presents the causal estimates for the ATT, ATN, and AOTT and their 95\% confidence intervals (CIs). We observe negative effects (reduction in sales) for individual and family-sized beverage sales in Philadelphia, but positive effects (increase in sales) in nearby counties for both sizes. In particular, the reduction in Philadelphia sales (-5291; 95\% CI: [-6003, -4579]) is nearly comparable to the increase in neighboring counties' sales (4693; 95\% CI: [2697, 6688]). However, we estimate that about half of the spillover effects (2334; 95\% CI: [1900, 2768]) would be returned to Philadelphia in the counterfactual scenario in which its neighboring counties were all treated, resulting in the AOTT estimate of -2959 (95\% CI: [-3888, -2026]). On the other hand, we estimate that roughly three fourths of the spillover effects would offset the ATT of the taxed family-sized beverages. All of these findings suggest significant effects of the excise tax on Philadelphia, which would have been evident even if its neighboring counties were also directly affected by the tax implementation.

\begin{table}[H]
\centering
\resizebox{0.9\textwidth}{!}{\begin{tabular}{llll}

   ATT & ATN & Offsetting effect & AOTT \\ 
  \hline
 \multicolumn{4}{l}{(1) Taxed individual sized beverages} \\  
   -5291  [-6003, -4579] &   4693 [2697, 6688] & -2334 [-1900, -2768] &  -2957 [-3888, -2026]\\
   \hline
\multicolumn{4}{l}{(2) Taxed family sized beverages}   \\
 -28416 [-31024, -25807] & 10886 [8945, 12832] &  -7518 [-6465, -8571] & -20898 [-23830, -17966]\\
 \hline
\end{tabular}}
\caption{\label{tab:bev_result}  Doubly-robust estimates (95\% confidence interval) for each causal estimand}
\end{table}

The Philadelphia beverage tax data present a few challenges. First, each store is de-identified, so each unit's actual distance to the border is unknown. We can only approximate the distance using zip code-level data (e.g., the shortest distance from the centroid of each zip code to the city border). Baseline covariates are also defined using census data at the zip-code level. Moreover, Assumption~\ref{as:link} may not be fully defendable since spillover behaviors are often dependent on \textit{post-} rather than \textit{pre-intervention} covariates (e.g., price changes after the intervention).
Still, our condition may be reasonable if consumers engage in cross-border purchasing after seeing announcements of the tax, regardless of the actual post-intervention price increases. For instance, in Philadelphia, there was widespread advertising of the tax on local TV channels, billboards, and in ads on the side of city buses.   
Lastly, we stratify retailers by type (supermarkets, mass merchandise stores, and pharmacies) and focus only on supermarket sales. However, it is possible that consumers are willing to change their purchasing habits following the intervention (e.g., switching from supermarkets to mass merchandise stores).
Despite these limitations, these data offer us a unique opportunity to infer diverse causal effects in the presence of spillover.

\section{Discussion}
\label{sec:discussion}


When projecting the effects of a proposed policy intervention, policy makers must consider how effects of the policy may spill over into neighboring regions. When spillover effects are prominent, differences in the neighboring policy status will have large ramifications on the effectiveness of the implemented policy. Therefore, policy makers may be interested in estimating what would have happened under an unobserved counterfactual neighborhood treatment assignment. In this work, we develop methods to estimate these counterfactual effects by using information on how a policy affected neighboring controls to infer what would have happened to the target population under a proposed neighborhood intervention status.
For example, our estimators may provide insights for policy makers hoping to use data from one region with a particular neighborhood intervention status (e.g., the beverage tax in Philadelphia with no treated neighbors) to project what would happen in a different intervention implementation (e.g., a beverage tax in Pittsburgh where some of the neighboring counties also implement the tax). Still, incorporating data from policy intervention implementations from different years, different policy specifications (e.g., tax rate), and varying population demographics will likely pose additional challenges. If multiple non-neighboring control groups that were neither directly nor indirectly affected by the policy at the time of interest are available, researchers could consider applying the synthetic control approach~\citep{abadie2010synthetic}.

There are a few limitations of our proposed estimators. First, we assume that observations are i.i.d. panel data. This assumption is likely to be violated when the outcomes (e.g., unit sales) are dependent on spatial factors not accounted for by the observed covariates. 
In addition, the outcome model assumes linearity in the error terms. Non-linear DiD methods (e.g., logistic regression) often fail to satisfy the counterfactual parallel trends assumptions (e.g., Assumptions~\ref{as:pt} and~\ref{as:pt2}). In these situations, one may consider an alternate way to specify the identification assumptions using a latent index model approach~\citep{athey2006identification, blundell2009alternative}.

In future work, we will investigate the generalizability and transportability of policy effects to a target population that may have different covariate compositions as well as different neighborhood treatment assignments. This direction will bring us one step closer to realizing the broader goal of developing causal inference methodologies to assist policy-makers  in designing and implementing future policy interventions that are effective for other -- possibly underrepresented -- populations.

\section*{Acknowledgement}
Youjin Lee, Gary Hettinger, and Nandita Mitra were supported by National Science Foundation (Award number: SES-2149716). We thank Drs.~Christina A. Roberto and Laura Gibson for providing insightful feedback on analyzing the Philadelphia beverage tax data.

\bibliographystyle{Chicago}
\bibliography{example}

\newpage
\spacingset{1.5}
\appendix

\setcounter{equation}{0}
\setcounter{figure}{0}
\setcounter{table}{0}
\setcounter{page}{1}
\setcounter{section}{0}
\setcounter{theorem}{0}
\renewcommand{\theequation}{S\arabic{equation}}
\renewcommand{\thetheorem}{S\arabic{theorem}}
\renewcommand{\thelemma}{S\arabic{lemma}}
\renewcommand{\thefigure}{S\arabic{figure}}
\renewcommand{\thesection}{S\arabic{section}}
\renewcommand{\thetable}{S\arabic{table}}
\renewcommand{\thefigure}{S\arabic{figure}}
\renewcommand{\theassumption}{S\arabic{assumption}}
\allowdisplaybreaks

\begin{center}
{\Large Supplementary Materials for ``Policy effect evaluation under counterfactual neighborhood interventions in the presence of spillover"}
\end{center}

\section{Proofs}
\label{sec:proofs}

\subsection{Proof of Lemma~\ref{lemma:ipw}}
\begin{proof}[Proof of Lemma~\ref{lemma:ipw}]
The following proof demonstrates that  $\hat{\delta}^{\text{ipw}}$ in Equation~\eqref{eq:delta} in the main text is a consistent estimator for $\delta := E(Y^{(1,1)}_{1} - Y^{(1,0)}_{1} \mid A = 1)$. Suppose that $\{\pi_{(a_{1}, a_{2})}(\bm{X})\}$ are correctly specified.
Then by the law of iterated expectations, we have the following:

\begin{eqnarray*}
\hat{\delta}^{\text{ipw}} & \overset{p}{\rightarrow} &
E \left\{  \frac{\pi_{(1,0)}(\bm{X})}{\pi_{(1,0)}} \left( \frac{I( g(\bm{A}) = (0,1) }{\pi_{(0,1)}(\bm{X})} - \frac{I( g(\bm{A}) = (0,0))}{\pi_{(0,0)}(\bm{X} ) } \right) (Y_{1} - Y_{0})  \right\} \\
&=& E\left[ E\left\{ \frac{\pi_{(1,0)}(\bm{X})}{\pi_{(1,0)}} \left( \frac{I( g(\bm{A}) = (0,1) }{\pi_{(0,1)}(\bm{X})} - \frac{I( g(\bm{A}) = (0,0))}{\pi_{(0,0)}(\bm{X}) } \right) (Y_{1} - Y_{0}) \mid \bm{X} \right\} \right] \\
&=& E\left[\frac{\pi_{(1,0)}(\bm{X})}{\pi_{(1,0)}} \left\{\frac{E\left( I(g(\bm{A}) = (0,1))(Y_{1} - Y_{0} ) \mid  \bm{X} \right)}{\pi_{(0,1)}(\bm{X})} -  \frac{E\left(  I(g(\bm{A}) = (0,0))(Y_{1} - Y_{0}) \mid \bm{X} \right) }{\pi_{(0,0)}(\bm{X})} \right\} \right] \\
&=&  E\left[\frac{\pi_{(1,0)}(\bm{X})}{\pi_{(1,0)}}\left\{ E\left( Y_{1} - Y_{0} \mid g(\bm{A}) = (0,1), \bm{X} \right)  - E\left(   Y_{1} - Y_{0} \mid g(\bm{A})=(0,0), \bm{X} \right) \right\}  \right] \\
&=&  E\left[\frac{\pi_{(1,0)}(\bm{X})}{\pi_{(1,0)}} \left\{ E( Y^{(0,1)}_{1} - Y^{(0,0)}_{0} \mid g(\bm{A}) = (0,1), \bm{X} ) - E(   Y^{(0,0)}_{1} - Y^{(0,0)}_{0} \mid g(\bm{A})=(0,0), \bm{X} )\right\}  \right] \\
&=&  E\left[\frac{\pi_{(1,0)}(\bm{X})}{\pi_{(1,0)}} \left\{ E( Y^{(0,1)}_{1} - Y^{(0,0)}_{0} \mid g(\bm{A}) = (0,1), \bm{X} ) - E(   Y^{(0,0)}_{1} - Y^{(0,0)}_{0} \mid g(\bm{A})=(0,1), \bm{X} )\right\}  \right] \\
&=&  E\left[\frac{\pi_{(1,0)}(\bm{X})}{\pi_{(1,0)}} E\left\{ Y^{(0,1)}_{1} - Y^{(0,0)}_{1} \mid g(\bm{A}) = (0,1), \bm{X} \right\} \right]  \\
&=&  E\left[\frac{\pi_{(1,0)}(\bm{X})}{\pi_{(1,0)}} E( Y^{(1,0)}_{1} - Y^{(1,1)}_{1} \mid A = 1, \bm{X} ) \right]  \\
&=&  E( Y^{(1,0)}_{1} - Y^{(1,1)}_{1} \mid A = 1 ).
\end{eqnarray*}

\end{proof}

\subsection{Proof of Theorem~\ref{thm:dr}}
\begin{proof}[Proof of Theorem~\ref{thm:dr}]
Suppose that $\Delta \hat{\mu}^{(0,1)}(\bm{X}) \overset{p}{\longrightarrow} E(Y_{1} - Y_{0} \mid g(\bm{A}) = (0,1), \bm{X})$ and  $\Delta \hat{\mu}^{(0,0)}(\bm{X}) \overset{p}{\longrightarrow}  E(Y_{1} - Y_{0} \mid g(\bm{A}) = (0,0), \bm{X})$ (i.e., correctly specified outcome model) as $n \rightarrow \infty$. Then the following results hold by the law of large numbers and the law of iterated expectations: 
\begin{eqnarray*}
 && E_{n}\left[  \frac{I(A = 1)}{\pi_{(1,0)}} \left\{ \Delta \hat{\mu}^{(0,1)}(\bm{X}) - \Delta \hat{\mu}^{(0,0)}(\bm{X}) \right\} \right] \\ 
 &\overset{p}{\longrightarrow}&  E\left[ \frac{I(A = 1)}{\pi_{(1,0)}} \left\{ E(Y^{(0,1)}_{1} - Y^{(0,0)}_{0} \mid g(\bm{A}) = (0,1), \bm{X}) - E(Y^{(0,0)}_{1} - Y^{(0,0)}_{0} \mid g(\bm{A}) = (0,0), \bm{X}) \right\}  \right] \\
 &=&  E\left[ \frac{I(A = 1)}{\pi_{(1,0)}} \left\{ E(Y^{(0,1)}_{1} - Y^{(0,0)}_{0} \mid g(\bm{A}) = (0,1), \bm{X}) - E(Y^{(0,0)}_{1} - Y^{(0,0)}_{0} \mid g(\bm{A}) = (0,1), \bm{X}) \right\}  \right] \\ 
 &=&  E\left[ \frac{I(A = 1)}{\pi_{(1,0)}} \left\{ E(Y^{(0,1)}_{1} - Y^{(0,0)}_{0} \mid g(\bm{A}) = (0,1), \bm{X}) - E(Y^{(0,0)}_{1} - Y^{(0,0)}_{0} \mid g(\bm{A}) = (0,1), \bm{X}) \right\}  \right] \\
 &=& E\left[ \frac{I(A = 1)}{\pi_{(1,0)}}  E(Y^{(0,1)}_{1} - Y^{(0,0)}_{1} \mid g(\bm{A}) = (0,1), \bm{X})    \right]
 \\ &=& E\left[ \frac{I(A_{i} = 1)}{\pi_{(1,0)}}  E(Y^{(1,1)}_{1} - Y^{(1,0)}_{1} \mid A = 1, \bm{X})    \right] \\ 
 &=& E(Y^{(1,1)}_{1} - Y^{(1,0)}_{1} \mid A = 1) \\
 &=& \delta.
\end{eqnarray*}
Therefore, it suffices to show that the remaining term in $\hat{\delta}^{\text{dr}}$ converges to zero. 
\begin{eqnarray*}
&& E_{n} \left[ \frac{\hat{\pi}_{(1,0)}(\bm{X}) I(g(\bm{A}) = (0,1)) }{\hat{\pi}_{(1,0)} \hat{\pi}_{(0,1)} (\bm{X})} \left\{ \Delta Y - \Delta \hat{\mu}^{(0,1)}(\bm{X}) \right\}  \right]  \\
&\overset{p}{\longrightarrow}& E\left[ E\left\{ \frac{\pi_{(1,0)}(\bm{X}) I(g(\bm{A}) = (0,1)) }{\pi_{(1,0)} \pi_{(0,1)} (\bm{X})} \left( \Delta Y - \Delta \mu^{(0,1)}(\bm{X}) \right)  \mid  \bm{X} \right\} \right] \\ 
&=& E\left[ \frac{\pi_{(1,0)}(\bm{X}) }{\pi_{(1,0)}} E \left( \Delta Y   \mid   g(\bm{A}) = (0,1), \bm{X} \right) - \frac{\pi_{(1,0)}(\bm{X}) }{ \pi_{(1,0)} \pi_{(0,1)}(\bm{X}) } \Delta \mu^{(0,1)}(\bm{X}) E(I(g(\bm{A}) = (0,1)) \mid \bm{X})    \right] \\ 
&=& E\left[ \frac{\pi_{(1,0)}(\bm{X}) }{\pi_{(1,0)}} E \left( \Delta Y   \mid   g(\bm{A}) = (0,1), \bm{X} \right) - \frac{\pi_{(1,0)}(\bm{X}) }{ \pi_{(1,0)} } \Delta \mu^{(0,1)}(\bm{X})    \right] \\  
&=& E\left[ \frac{\pi_{(1,0)}(\bm{X})}{\pi_{(1,0)}} \left\{ E(\Delta Y \mid g(\bm{A}) = (0, 1), \bm{X}) - E(Y_{1} - Y_{0} \mid g(\bm{A}) = (0, 1), \bm{X}) \right\} \right] \\ 
&=& 0. 
\end{eqnarray*}
Similarly, we can show that
\begin{eqnarray*}
&& E_{n} \left[ \frac{\hat{\pi}_{(1,0)}(\bm{X}) I(g(\bm{A}) = (0,0)) }{\hat{\pi}_{(1,0)} \hat{\pi}_{(0,0)} (\bm{X})} \left\{ \Delta Y - \Delta \hat{\mu}^{(0,0)}(\bm{X}) \right\}   \right]  \\
 &\overset{p}{\longrightarrow}& E\left[ E\left\{ \frac{\pi_{(1,0)}(\bm{X}) I(g(\bm{A}) = (0,0)) }{\pi_{(1,0)} \pi_{(0,0)} (\bm{X})} \left( \Delta Y - \Delta \mu^{(0,0)}(\bm{X}) \right)  \mid  \bm{X} \right\} \right] \quad \text{ as } n \rightarrow \infty \\
&=& 0.
\end{eqnarray*}
Therefore, $\hat{\delta}^{\text{dr}} \overset{p}{\longrightarrow} \delta$ when the outcome regression models are correctly specified.

Now suppose that the propensity score model is correctly specified. That is, when $\hat{\pi}_{(a_{1},a_{2})}(\bm{X}) \overset{p}{\longrightarrow} Pr(I(g(\bm{A})) = (a_{1}, a_{2}) \mid \bm{X})$ for $(a_{1}, a_{2}) \in \{(1,0), (0,1), (0,0)\}$.
Then the following results hold.

\begin{eqnarray*}
&&   E_{n} \left\{ \left( \frac{\hat{\pi}_{(1,0)}(\bm{X}) I(g(\bm{A}) = (0,1)) }{ \hat{\pi}_{(1,0)} \hat{\pi}_{(0,1)}(\bm{X}) } - \frac{I(A = 1)}{\hat{\pi}_{(1,0)}} \right)  \Delta \hat{\mu}^{(0,1)} (\bm{X}) \right\}   \\
&\overset{p}{\longrightarrow} & E \left[ E \left\{ \left( \frac{\pi_{(1,0)}(\bm{X}) I(g(\bm{A}) = (0,1)) }{ \pi_{(1,0)} \pi_{(0,1)}(\bm{X}) } - \frac{I(A = 1)}{\pi_{(1,0)}} \right)  \Delta \mu^{(0,1)} (\bm{X}) \mid \bm{X} \right\} \right] 
\\ &=&  E \left[ \Delta \mu^{(0,1)}(\bm{X}) \left\{ \frac{\pi_{(1,0)}(\bm{X})\pi_{(0,1)}(\bm{X}) }{\pi_{(1,0)} \pi_{(0,1)}(\bm{X}) } - \frac{Pr(I(A_{i} = 1) \mid \bm{X}) }{\pi_{(1,0)}}  \right\} \right]  \\ 
&=& 0.
\end{eqnarray*}
Similarly, we can show that 
\begin{eqnarray*}
 E_{n} \left[ \left\{ \frac{\hat{\pi}_{(1,0)}(\bm{X}) I(g(\bm{A}) = (0,0)) }{ \hat{\pi}_{(1,0)} \hat{\pi}_{(0,0)}(\bm{X}) } - \frac{I(A = 1)}{\hat{\pi}_{(1,0)}} \right\}  \Delta \hat{\mu}^{(0,0)} (\bm{X})  \right]  \overset{p}{\longrightarrow} 0.
\end{eqnarray*}
Therefore, 
\begin{eqnarray*}
\hat{\delta}^{\text{dr}} &-& E_{n}\left[   \left\{\frac{\hat{\pi}_{(1,0)}(\bm{X}) I(g(\bm{A}) = (0,1)) }{ \hat{\pi}_{(1,0)} \hat{\pi}_{(0,1)}(\bm{X}) } - \frac{I(A = 1)}{\hat{\pi}_{(1,0)}} \right\} \Delta \hat{\mu}^{(0,1)} (\bm{X}) \right]\\
&-& E_{n}\left[   \left\{ \frac{\hat{\pi}_{(1,0)}(\bm{X}) I(g(\bm{A}) = (0,0)) }{ \hat{\pi}_{(1,0)} \hat{\pi}_{(0,0)}(\bm{X}) } - \frac{I(A = 1)}{\hat{\pi}_{(1,0)}} \right\}  \Delta \hat{\mu}^{(0,0)} (\bm{X}) \right] \\ 
&=& \hat{\delta}^{\text{ipw}}.
\end{eqnarray*}
Lemma~\ref{lemma:ipw} shows that $\hat{\delta}^{\text{ipw}} \overset{p}{\rightarrow} \delta$. This concludes that $\hat{\delta}^{\text{dr}}$ is a consistent estimator for $\delta$ when the propensity score model is correctly specified.

\end{proof}

\subsection{Proof of Theorem~\ref{thm:variance}}
\begin{proof}[Proof of Theorem~\ref{thm:variance}]

By Assumptions Assumptions~\ref{as:mapping}--\ref{as:pos} and~\ref{as:link}--\ref{as:pt2}, we can represent the target estimand as follows: 
\begin{eqnarray*}
\delta &=& \int\{ E(Y_{1} - Y_{0} \mid g(\bm{A}) = (0,1), \bm{X} = \bm{x}) - E(Y_{1} - Y_{0} \mid g(A) = (0,0), \bm{X} = \bm{x}) \} \\ && \quad dP(\bm{x} \mid g(\bm{A}) = (1,0)) \\ 
&=& \int \int \{ (y_{1} - y_{0}) dP(y_{1}, y_{0} \mid g(\bm{A}) = (0,1), \bm{X} = \bm{x}) \} dP(\bm{x} \mid g(\bm{A}) = (1,0)) \\ 
&& - \int\int \{ (y_{1}-y_{0}) dP(y_{1}, y_{0} \mid g(\bm{A}) = (0,0), \bm{X} = \bm{x}) \} dP(\bm{x} \mid g(\bm{A}) = (1,0)).
\end{eqnarray*}
Our observed data are i.i.d. $O_{i} = \{Y_{i1}, Y_{i0}, g_{i}(\bm{A}), \bm{X}_{i}\}_{i=1}^{n}$ from some unknown probability distribution $P_{0}$. Suppose that the distribution $P_{0}$ has the following density: 
\begin{eqnarray*}
p(O) &=& p(y_{1}, y_{0} \mid g(\bm{A}) = (1,0), \bm{x})^{I(g(\bm{A}) = (1,0))} \pi_{(1,0)}(\bm{x})^{I(g(\bm{A}) = (1,0))} \\
&& \times p(y_{1}, y_{0} \mid g(\bm{A}) = (0, 1), \bm{x})^{I(g(\bm{A}) = (0,1))} \pi_{(0,1)}(\bm{x})^{I(g(\bm{A}) = (0,1))} \\
&& \times p(y_{1}, y_{0} \mid g(\bm{A}) = (0,0), \bm{x})^{I(g(\bm{A})) = (0,0)} \pi_{(0,0)}(\bm{x})^{I(g(\bm{A})) = (0,0)} \times p(\bm{x}).
\end{eqnarray*}

Let $p_{\epsilon}(O) = p(O;\epsilon)$ denote a parametric submodel with parameter $\epsilon \in \mathbb{R}$. Let $l (y \mid x;\epsilon) = \log p(y \mid x;\epsilon)$ for any $(X,Y) \in O$ and $l^{\prime}_{\epsilon}(y; \epsilon) = \partial l(y;\epsilon) / \partial \epsilon$. 
Then the corresponding score of the above density is given by: 
\begin{eqnarray*}
l^{\prime}_{\epsilon}(O; \epsilon) &=& I(g(\bm{A}) = (1,0))  \{ l^{\prime}_{\epsilon} (y_{1}, y_{0} \mid g(\bm{A}) = (1,0), \bm{x}; \epsilon) + l^{\prime}_{\epsilon}(g(\bm{A}) = (1,0) \mid \bm{x}; \epsilon) \}\\
&& + I(g(\bm{A}) = (0,1)) \{ l^{\prime}_{\epsilon} (y_{1}, y_{0} \mid g(\bm{A}) = (0, 1), \bm{x}; \epsilon) +  l^{\prime}_{\epsilon}(g(\bm{A}) = (0,1) \mid \bm{x}; \epsilon) \} \\
&& + I(g(\bm{A})) = (0,0) \{ l^{\prime}_{\epsilon}(y_{1}, y_{0} \mid g(\bm{A}) = (0,0), \bm{x}; \epsilon) + l^{\prime}_{\epsilon}(g(\bm{A}) = (0,0) \mid \bm{x}; \epsilon) \} + l^{\prime}_{\epsilon}(\bm{x}; \epsilon).
\end{eqnarray*}

\begin{eqnarray*}
\delta^{\prime}_{\epsilon}(\epsilon) &=& 
\int\int \{ (y_{1} - y_{0}) l^{\prime}_{\epsilon}(y_{1}, y_{0} \mid g(\bm{A}) = (0, 1), \bm{x}; \epsilon) dP(y_{1}, y_{0} \mid g(\bm{A})=(0,1), \bm{x}; \epsilon)  \\ 
&&  - (y_{1}-y_{0}) l^{\prime}_{\epsilon}(y_{1}, y_{0} \mid g(\bm{A}) = (0,0), \bm{x}; \epsilon) dP(y_{1}, y_{0} \mid g(\bm{A}) = (0,0), \bm{x}; \epsilon) \}dP(\bm{x} \mid g(\bm{A}) = (1,0); \epsilon) \\
&& + \int\int \{  (y_{1} - y_{0}) dP(y_{1}, y_{0} \mid g(\bm{A}) = (0,1), \bm{x}; \epsilon) \\
&&  - (y_{1}-y_{0}) dP(y_{1}, y_{0} \mid g(\bm{A}) = (0,0), \bm{x}; \epsilon) \} l^{\prime}_{\epsilon}(\bm{x} \mid g(\bm{A}) = (1,0); \epsilon) dP(\bm{x} \mid g(\bm{A}) = (1,0); \epsilon)
\\ 
&=& 
\int\int \{ (y_{1} - y_{0}) l^{\prime}_{\epsilon}(y_{1}, y_{0} \mid g(\bm{A}) = (0, 1), \bm{x}; \epsilon) dP(y_{1}, y_{0} \mid g(\bm{A})=(0,1), \bm{x}; \epsilon)  \\ 
&&  - (y_{1}-y_{0}) l^{\prime}_{\epsilon}(y_{1}, y_{0} \mid g(\bm{A}) = (0,0), \bm{x}; \epsilon) dP(y_{1}, y_{0} \mid g(\bm{A}) = (0,0), \bm{x}; \epsilon) \}dP(\bm{x} \mid g(\bm{A}) = (1,0); \epsilon) \\
&& + \int\int \{  (y_{1} - y_{0}) dP(y_{1}, y_{0} \mid g(\bm{A}) = (0,1), \bm{x}; \epsilon)  - (y_{1}-y_{0}) dP(y_{1}, y_{0} \mid g(\bm{A}) = (0,0), \bm{x}; \epsilon) \}  
\\ && \quad \times \{ l^{\prime}_{\epsilon}(g(\bm{A}) = (1,0) \mid \bm{X} = \bm{x}; \epsilon) + l^{\prime}_{\epsilon}(\bm{x}; \epsilon)  \} dP(\bm{x} \mid g(\bm{A}) = (1,0); \epsilon).
\end{eqnarray*}
The last equality uses
\begin{eqnarray*}
dP(\bm{x} \mid g(\bm{A}) = (1,0)) =  \frac{dP(g(\bm{A}) = (1,0) \mid \bm{X} = \bm{x}) dP(\bm{X} = \bm{x}) }{dP(g(\bm{A}) = (1,0))}  ,
\end{eqnarray*}
so that
\begin{eqnarray*}
&& l^{\prime}_{\epsilon}(\bm{x} \mid g(\bm{A}) = (1,0); \epsilon) dP(\bm{x} \mid g(\bm{A}) = (1,0); \epsilon) \\ &=& \frac{\partial}{ \partial \epsilon}\left\{ \frac{dP(g(\bm{A}) = (1,0) \mid  \bm{x}; \epsilon ) dP(\bm{x}; \epsilon)}{ dP(g(\bm{A}) = (1,0))}  \right\} \\
&=& \frac{dP^{\prime}_{\epsilon}(g(\bm{A}) = (1,0) \mid \bm{x}; \epsilon)dP(\bm{x}; \epsilon)  + dP(g(\bm{A}) = (1,0) \mid \bm{x}; \epsilon) dP^{\prime}_{\epsilon}(\bm{x}) }{   dP(g(\bm{A}) = (1,0))  }    
\\ &=& \{ l^{\prime}_{\epsilon}(g(\bm{A}) = (1,0) \mid \bm{x}; \epsilon)+ l^{\prime}_{\epsilon}(\bm{x};\epsilon) \} \frac{ dP(g(\bm{A} = (1,0) \mid \bm{x}; \epsilon)) dP(\bm{x}; \epsilon) }{dP(g(\bm{A}) = (1,0))} \\
&=& \{ l^{\prime}_{\epsilon}(g(\bm{A}) = (1,0) \mid \bm{x}; \epsilon)+ l^{\prime}_{\epsilon}(\bm{x}; \epsilon) \} dP(\bm{x} \mid g(\bm{A}) = (1,0); \epsilon)
.
\end{eqnarray*}
We now demonstrate that the proposed influence function $\psi$ is a pathwise derivative, i.e., $E(\psi l^{\prime}_{\epsilon}(O; \epsilon = 0)) = \delta^{\prime}_{\epsilon}(\epsilon = 0)$. 
First,
\begin{eqnarray*}
&& E\left[ \frac{\pi_{(1,0)}(\bm{X})}{\pi_{(1,0)}} \frac{I(g(\bm{A})=(0,1))}{\pi_{(0,1)}(\bm{X})} (Y_{1} - Y_{0}) l^{\prime}_{\epsilon}(O; \epsilon)\right] \\
&=& E\left[ \frac{\pi_{(1,0)}(\bm{X})}{\pi_{(1,0)}} \frac{I(g(\bm{A})=(0,1))}{\pi_{(0,1)}(\bm{X})} (Y_{1} - Y_{0}) l^{\prime}_{\epsilon}(y_{1}, y_{0} \mid g(\bm{A}) = (0,1), \bm{x}; \epsilon)\right] \\
&=& E\left[ \frac{\pi_{(1,0)}(\bm{X})}{\pi_{(1,0)}} E\left\{ (Y_{1} - Y_{0}) l^{\prime}_{\epsilon}(y_{1}, y_{0} \mid g(\bm{A}) = (0,1), \bm{x}; \epsilon) \mid g(\bm{A}) = (0, 1), \bm{x} \right\} \right]  \\ 
&=& \int \int (y_{1} - y_{0}) l^{\prime}_{\epsilon}(y_{1}, y_{0} \mid g(\bm{A}) = (0,1), \bm{x}; \epsilon) dP(y_{1}, y_{0} \mid g(\bm{A}) = (0, 1), \bm{x}) \frac{\pi_{(1,0)}(\bm{x})}{\pi_{(1,0)}} dP(\bm{x}) \\ 
&=& \int \int (y_{1} - y_{0}) l^{\prime}_{\epsilon}(y_{1}, y_{0} \mid g(\bm{A}) = (0,1), \bm{x}; \epsilon) dP(y_{1}, y_{0} \mid g(\bm{A}) = (0, 1), \bm{x}) dP(\bm{x} \mid g(\bm{A}) = (1,0)).
\end{eqnarray*}
Similarly, we can show that 
\begin{eqnarray*}
&& E\left[ \frac{\pi_{(1,0)}(\bm{X})}{\pi_{(1,0)}} \frac{I(g(\bm{A})=(0,0))}{\pi_{(0,0)}(\bm{X})} (Y_{1} - Y_{0}) l^{\prime}_{\epsilon}(O; \epsilon)\right] \\
&=& \int \int (y_{1} - y_{0}) l^{\prime}_{\epsilon}(y_{1}, y_{0} \mid g(\bm{A}) = (0,0), \bm{x}; \epsilon) dP(y_{1}, y_{0} \mid g(\bm{A}) = (0, 0), \bm{x}) dP(\bm{x} \mid g(\bm{A}) = (1,0); \epsilon).
\end{eqnarray*}

Second,
\begin{eqnarray*}
&&E\left[ \frac{\pi_{(1,0)}(\bm{X})}{\pi_{(1,0)}} \frac{I(g(\bm{A})=(0,1))}{\pi_{(0,1)}(\bm{X})} \Delta \mu^{(0,1)}(\bm{X}) l^{\prime}_{\epsilon}(O; \epsilon) \right] \\ 
&=& E\left[   \frac{\pi_{(1,0)}(\bm{X})}{\pi_{(1,0)}} \frac{I(g(\bm{A})=(0,1))}{\pi_{(0,1)}(\bm{X})} \Delta \mu^{(0,1)}(\bm{X}) l^{\prime}_{\epsilon}(g(\bm{A}) = (0,1) \mid \bm{x}; \epsilon)\right]
\\ &=& \int \Delta \mu^{(0,1)}(\bm{x}) \{  l^{\prime}_{\epsilon}(g(\bm{A}) = (1,0) \mid \bm{x}; \epsilon) + l^{\prime}_{\epsilon}(\bm{x};\epsilon) \} dP(\bm{x} \mid  g(\bm{A}) = (1,0)).
\end{eqnarray*}
Similarly, 
\begin{eqnarray*}
&&E\left[ \frac{\pi_{(1,0)}(\bm{X})}{\pi_{(1,0)}} \frac{I(g(\bm{A})=(0,0))}{\pi_{(0,0)}(\bm{X})} \Delta \mu^{(0,0)}(\bm{X}) l^{\prime}_{\epsilon}(O; \epsilon) \right] 
\\ &=& \int \Delta \mu^{(0,0)}(\bm{x}) \{  l^{\prime}_{\epsilon}(g(\bm{A}) = (1,0) \mid \bm{x}; \epsilon) + l^{\prime}_{\epsilon}(\bm{x};\epsilon) \} dP(\bm{x} \mid  g(\bm{A}) = (1,0)).
\end{eqnarray*}
Since 
$E\left[ \frac{I(A_{i} = (1,0) )}{\pi_{(1,0)} } \{ \Delta \mu^{(0,1)}(\bm{X}) - \Delta \mu^{(0,0)} (\bm{X}) -\delta \} l^{\prime}_{\epsilon}(O;\epsilon) \right] = 0$, we can conclude that
\begin{eqnarray*}
&& E[\psi_{\delta}(Y_{1}, Y_{0}, \bm{A}, \bm{X})  l^{\prime}_{\epsilon}(O; \epsilon)] \\ 
&=& \int \int (y_{1} - y_{0}) l^{\prime}_{\epsilon}(y_{1}, y_{0} \mid g(\bm{A}) = (0,1), \bm{x}; \epsilon) dP(y_{1}, y_{0} \mid g(\bm{A}) = (0, 1), \bm{x}) dP(\bm{x} \mid g(\bm{A}) = (1,0); \epsilon) \\
&-& \int \int (y_{1} - y_{0}) l^{\prime}_{\epsilon}(y_{1}, y_{0} \mid g(\bm{A}) = (0,0), \bm{x}; \epsilon) dP(y_{1}, y_{0} \mid g(\bm{A}) = (0, 0), \bm{x}) dP(\bm{x} \mid g(\bm{A}) = (1,0); \epsilon) \\
&-& \int \Delta \mu^{(0,1)}(\bm{x}) \{  l^{\prime}_{\epsilon}(g(\bm{A}) = (1,0) \mid \bm{x}; \epsilon) + l^{\prime}_{\epsilon}(\bm{x};\epsilon) \} dP(\bm{x} \mid  g(\bm{A}) = (1,0))\\
&+& \int \Delta \mu^{(0,0)}(\bm{x}) \{  l^{\prime}_{\epsilon}(g(\bm{A}) = (1,0) \mid \bm{x}; \epsilon) + l^{\prime}_{\epsilon}(\bm{x};\epsilon) \} dP(\bm{x} \mid  g(\bm{A}) = (1,0)) 
\\ &=& \delta^{\prime}_{\epsilon}(\epsilon).
\end{eqnarray*}
That is, $\delta$ is pathwise differentiable, and $\psi_{\delta} (Y_{1}, Y_{0}, \bm{A}, \bm{X})$ is on the tangent space of a parametric submodel. Therefore, the semiparametric efficiency bound is given by the expected value of $\psi^{2}_{\delta}(Y_{1}, Y_{0}, \bm{A}, \bm{X})$~\citep{newey1990semiparametric, kennedy2016semiparametric}.

\end{proof}

\subsection{Proof of Theorem~\ref{thm:asymptotic}}
\begin{proof}[Proof of Theorem~\ref{thm:asymptotic}]

Assume that $\hat{\pi}_{(1,0)}(\bm{X}; \alpha^0)$ is a parametric model with the finite dimensional ``pseudo-true" parameter $\alpha^{0}$. Similarly define $\hat{\pi}_{(0,1)}(\bm{X}; \alpha^{0})$ and $\hat{\pi}_{(0,0)}(\bm{X}; \alpha^{0})$. Let $\hat{\mu}^{(0,1)}_{t}(\bm{X}, \beta^{0})$ and $\hat{\mu}^{(0,0)}_{t}(\bm{X}, \beta^{0})$ be parametric models for $\mu^{(0,1)}_{t} (\bm{X})$ and $\mu^{(0,0)}(\bm{X})$, respectively ($t$=0,1). 
Then $\hat{{\alpha}}$ and $\hat{{\beta}}$ are estimators for ${\alpha}^{0}$ and ${\beta}^{0}$, respectively. Consider the following standard conditions regarding smoothness of parametric models and asymptotically linear representations. Specifically, we first assume that parametric models $\{\hat{\pi}_{(a_{1}, a_{2})} (\bm{X})\}$ and $\{\hat{\mu}^{(a_{1}, a_{2})}_{t} (\bm{X})\}$ are smooth and twice continuously differentiable around the true values of $(\alpha^{0}, \beta^{0})$. 
For simplicity in notation, we omit $\hat{(\cdot)}$ in $\hat{\pi}_{(a_{1}, a_{2})}(\bm{X}; \alpha^{0})$ and $\hat{\mu}^{(a_{1}, a_{2})}_{t}(\bm{X}; \beta^{0})$ for $(a_{1}, a_{2}) \in \{ (1,0), (0,1), (0,0) \}$. Moreover, we assume that there exist functions $l_{\pi}(;\alpha)$ and $l_{\mu}(;\beta)$ with mean zero and finite variances such that:
\begin{eqnarray*}
\sqrt{n}(\hat{\alpha} - \alpha^{0}) &=& \frac{1}{\sqrt{n}} \sum\limits_{i=1}^{n} l_{\pi}(\bm{A}_{i}, \bm{X}_{i}; \alpha^{0}) + o_{P}(1) \\
\sqrt{n}(\hat{\beta} - \beta^{0}) &=& \frac{1}{\sqrt{n}}\sum\limits_{i=1}^{n} l_{\mu}(O_{i}; \beta^{0}) + o_{P}(1).
\end{eqnarray*}

Define the following notation:
\begin{eqnarray*}
w_{0}(\bm{A}_{i}) &=& \frac{I(A_{i} = 1)}{\pi_{(1,0)}} \\
w_{1}(\bm{A}_{i}, \bm{X}_{i}; \alpha^{0}) &=& \frac{\pi_{(1,0)} (\bm{X}_{i}; \alpha^{0}) I(g_{i}(\bm{A}) = (0,1))}{ \pi_{(1,0)} \pi_{(0,1)}(\bm{X}_{i}; \alpha^{0}) } \\
w_{2}(\bm{A}_{i}, \bm{X}_{i}; \alpha^{0}) &=& \frac{\pi_{(1,0)} (\bm{X}_{i}; \alpha^{0}) I(g_{i}(\bm{A}) = (0,0))}{ \pi_{(1,0)} \pi_{(0,0)}(\bm{X}_{i}; \alpha^{0}) }\\ 
\dot{\mu}^{(a_{1}, a_{2})}_{t}(\bm{X}_{i}; \beta^{0}) &=&  \frac{\partial \mu^{(a_{1}, a_{2})}_{t}(\bm{X}_{i}; \beta^{0})}{\partial \beta^{0}} \quad \text{ for } (a_{1}, a_{2}) \in \{(0,0), (0,1)\};~ t= 0,1\\ 
\dot{w}_{1}(\bm{A}_{i}, \bm{X}_{i}; \alpha^{0}) &=& \frac{I(g_{i}(\bm{A}) = (0,1))}{\pi_{(1,0)}} \left\{  \frac{\dot{\pi}_{(1,0)}(\bm{X}_{i}; \alpha^{0}) \pi_{(0,1)}(\bm{X}_{i}; \alpha^{0}) - \pi_{(1,0)}(\bm{X}_{i};\alpha^{0}) \dot{\pi}_{(0,1)}(\bm{X}_{i}; \alpha^{0}) }{\pi^2_{(0,1)} (\bm{X}_{i}; \alpha^{0}) }\right\} \\ 
\dot{w}_{2}(\bm{A}_{i}, \bm{X}_{i}; \alpha^{0}) &=& \frac{I(g_{i}(\bm{A}) = (0,1))}{\pi_{(1,0)}} \left\{  \frac{\dot{\pi}_{(1,0)}(\bm{X}; \alpha^{0}) \pi_{(0,0)}(\bm{X}_{i}; \alpha^{0}) - \pi_{(1,0)}(\bm{X}_{i};\alpha^{0}) \dot{\pi}_{(0,0)}(\bm{X}_{i}; \alpha^{0}) }{\pi^2_{(0,0)} (\bm{X}_{i}; \alpha^{0}) }\right\},
\end{eqnarray*}
where $\dot{\pi}_{(a_{1}, a_{2})}(\bm{X}_{i}; \alpha^{0})$ is the derivative of $\pi_{(a_{1}, a_{2})}(\bm{X}_{i}; \alpha^{0})$ with respect to $\alpha^{0}$.

Following the proofs of Theorem A.1 in~\cite{sant2020doubly} that use a second-order Taylor expansion around the true value of $\alpha^{0}$ and $\beta^{0}$, the following equations hold.
\begin{eqnarray*}
&& \sqrt{n}(\hat{\delta}^{dr} - \delta) 
\\ &=& \frac{1}{\sqrt{n}} \sum\limits_{i=1}^{n} w_{1}(\bm{A}_{i},\bm{X}_{i}; \alpha^{0})\{ \Delta Y_{i} - E(w_{1}(\bm{A}, \bm{X}; \alpha^{0}) \Delta Y_{i} )  \} \\ 
&-&   w_{2}(\bm{A}_{i},\bm{X}_{i}; \alpha^{0})\{ \Delta Y_{i} - E(w_{2}(\bm{A}, \bm{X}; \alpha^{0}) \Delta Y )  \} \\ 
&+& l_{\pi}(\bm{A}_{i}, \bm{X}_{i}; \alpha^{0})^{\prime} E( \dot{w}_{1}(\bm{A}, \bm{X}; \alpha^{0}) (\Delta Y - E(w_{1}(\bm{A}, \bm{X}; \alpha^{0}) \Delta Y )) )  \\
&-&l_{\pi}(\bm{A}_{i}, \bm{X}_{i}; \alpha^{0})^{\prime} E( \dot{w}_{2}(\bm{A}, \bm{X}; \alpha^{0}) (\Delta Y - E(w_{2}(\bm{A}, \bm{X}; \alpha^{0}) \Delta Y )) ) \\
&-& w_{1}(\bm{A}_{i}, \bm{X}_{i}; \alpha^{0}) (\Delta \mu^{(0,1)}(\bm{X}_{i}; \beta^{0}) - E(w_{1}(\bm{A}, \bm{X}; \alpha^{0}) \Delta \mu^{(0,1)}(\bm{X}; \beta^{0}) ) ) \\ &-& l_{\pi}(\bm{A}_{i}, \bm{X}_{i}; \alpha^{0})^{\prime} E( \dot{w}_{1}(\bm{A}, \bm{X}; \alpha^{0}) ( \Delta \mu^{(0,1)}(\bm{X}; \beta^{0}) - E( w_{1}(\bm{A}, \bm{X}; \alpha^{0}) \Delta \mu^{(0,1)}(\bm{X}; \beta^{0})  )  )    )\\ 
&-& l_{\mu}(O_{i}; \beta^{0})^{\prime} E(w_{1}(\bm{A}, \bm{X}; \alpha^{0}) \Delta \dot{\mu}^{(0,1)}(\bm{X}; \beta^{0}) )  \\ 
&+& w_{2} (\bm{A}_{i}, \bm{X}_{i}; \alpha^{0}) ( \Delta \mu^{(0,0)}(\bm{X}_{i}; \beta^{0}) - E(w_{2}(\bm{A}_{i}, \bm{X}_{i}; \alpha^{0}) \Delta \mu^{(0,0)}(\bm{X}_{i}; \beta^{0}) ) )
\\ &-& l_{\pi}(\bm{A}_{i}, \bm{X}_{i}; \alpha^{0}) E( \dot{w}_{2}(\bm{A}, \bm{X}; \alpha^{0}) ( \Delta \mu^{(0,0)}(\bm{X}; \beta^{0}) - E( w_{2}(\bm{A}, \bm{X}; \alpha^{0}) \Delta \mu^{(0,1)}(\bm{X}; \beta^{0})  )  )    ) \\ 
&-& l_{\mu}(O_{i}; \beta^{0}) E(w_{2}(\bm{A}, \bm{X}; \alpha^{0}) \Delta \dot{\mu}^{(0,0)}(\bm{X}; \beta^{0}) )  \\ 
&+& w_{0}(\bm{A}_{i}) ( \Delta \mu^{(0,1)}(\bm{X}_{i}; \beta^{0}) - E( w_{0}(\bm{A}) \Delta \mu^{(0,1)} (\bm{X}; \beta^{0}) )  ) \\
&+& l_{\mu}(O_{i}; \beta^{0})^{\prime} E(  w_{0}(\bm{A}) \Delta \dot{\mu}^{(0,1)} (\bm{X}; \beta^{0}) )  ) \\
&-& w_{0}(\bm{A}_{i}) ( \Delta \mu^{(0,0)}(\bm{X}_{i}; \beta^{0}) - E( w_{0}(\bm{A}) \Delta \mu^{(0,0)} (\bm{X}; \beta^{0}) )  ) \\
&-& l_{\mu}(O_{i}; \beta^{0})^{\prime} E(  w_{0}(\bm{A}) \Delta \dot{\mu}^{(0,0)} (\bm{X}; \beta^{0}) )  ) + o_{P}(1)
\\ &=& \frac{1}{\sqrt{n}} \sum\limits_{i=1}^{n} w_{1}(\bm{A}_{i}, \bm{X}_{i}; \alpha^{0}) \{  \Delta Y_{i} - \Delta \mu^{(0,1)}(\bm{X}_{i}; \beta^{0}) - E( w_{1}(\bm{A}, \bm{X}; \alpha^{0}) ( \Delta Y - \Delta \mu^{(0,1)}(\bm{X}; \beta^{0}) ) )    \} \\ 
&-&  w_{2}(\bm{A}_{i}, \bm{X}_{i}; \alpha^{0}) \{  \Delta Y_{i} - \Delta \mu^{(0,0)}(\bm{X}_{i}; \beta^{0}) - E( w_{2}(\bm{A}, \bm{X}; \alpha^{0}) ( \Delta Y - \Delta \mu^{(0,0)}(\bm{X}; \beta^{0}) ) )    \} \\
&+& w_{0}(\bm{A}_{i}) \{ \Delta \mu^{(0,1)}(\bm{X}_{i}; \alpha^{0}) - \Delta \mu^{(0,0)}(\bm{X}_{i}; \alpha^{0}) - E(  w_{0}(\bm{A}) (\Delta \mu^{(0,1)}(\bm{X}; \alpha^{0}) - \Delta \mu^{(0,0)}(\bm{X}; \alpha^{0}))  ) \} \\
&+& l_{\pi}(\bm{A}_{i}, \bm{X}_{i}; \alpha^{0})^{\prime} E[  \dot{w}_{1}(\bm{A}, \bm{X}; \alpha^{0}) \cdot \{\Delta Y - \Delta \mu^{(0,1)}(\bm{X}; \beta^{0}) - E( w_{1}(\bm{A}, \bm{X}; \alpha^{0}) ( \Delta Y - \Delta \mu^{(0,1)}(\bm{X}; \beta^{0})) ) \} ] \\
&-& l_{\pi}(\bm{A}_{i}, \bm{X}_{i}; \alpha^{0})^{\prime} E[  \dot{w}_{2}(\bm{A}, \bm{X}; \alpha^{0}) \cdot \{\Delta Y - \Delta \mu^{(0,0)}(\bm{X}; \beta^{0}) - E( w_{2}(\bm{A}, \bm{X}; \alpha^{0}) ( \Delta Y - \Delta \mu^{(0,0)}(\bm{X}; \beta^{0})) ) \} ] \\
&-& l_{\mu}(O_{i}; \beta^{0})^{\prime} E\{ ( w_{1}(\bm{A}, \bm{X}; \alpha^{0}) - w_{0}(\bm{A}) )\Delta \dot{\mu}^{(0,1)}(\bm{X}; \beta^{0}) \} \\ 
&+&  l_{\mu}(O_{i}; \beta^{0})^{\prime} E\{ ( w_{2}(\bm{A}, \bm{X}; \alpha^{0}) - w_{0}(\bm{A}) )\Delta \dot{\mu}^{(0,0)}(\bm{X}; \beta^{0}) \} + o_{P}(1)
\end{eqnarray*}

Now suppose that both the propensity score working model and the outcome working model are correctly specified. That is, $\pi_{(a_{1}, a_{2})}(\bm{X}; \alpha^{0}) = \pi_{(a_{1}, a_{2})}(\bm{X})$ and $\Delta \mu^{(a_{1}, a_{2})} (\bm{X}; \beta^{0}) = \Delta \mu^{(a_{1}, a_{2})}(\bm{X})$ for $(a_{1}, a_{2}) \in \{(1,0), (0,1), (0,0)\}$. Then it follows that
\begin{eqnarray}
\label{eq:aux1}
&& E[w_{1}(\bm{A}, \bm{X}; \alpha^{0}) (\Delta Y - \Delta \mu^{(0,1)}(\bm{X};\beta^{0})) ] \nonumber \\ 
&=& E \left[ \frac{\pi_{(1,0)}(\bm{X}) I(g(\bm{A} = (0,1))) }{\pi_{(1,0)} \pi_{(0,1)}(\bm{X})} \{ \Delta Y - \Delta \mu^{(0,1)}(\bm{X}) \}  \right] \\
&=& E \left[  \frac{\pi_{(1,0)}(\bm{X})}{\pi_{(1,0)}} \left\{  E(\Delta Y \mid g(\bm{A}) = (0,1), \bm{X}) - \Delta \mu^{(0,1)} (\bm{X})  \right\} \right] \nonumber \\ 
&=& 0. \nonumber
\end{eqnarray}
Similarly, $E[w_{2}(\bm{A}, \bm{X}; \alpha^{0}) (\Delta Y - \Delta \mu^{(0,0)}(\bm{X};\beta^{0})) ] = 0$ when both models are correctly specified. 
\begin{eqnarray*}
&& E[ (w_{1} (\bm{A}, \bm{X}) - w_{0}(\bm{A}) ) \Delta \dot{\mu}^{(0,1)}(\bm{X}) ] \\ 
&=& E\left[ \frac{1}{\pi_{(1,0)}} \left\{ \frac{\pi_{(1,0)}(\bm{X}) I(g(\bm{A}) = (0,1) )   }{\pi_{(0,1)}(\bm{X})} - I(A = 1)  \right\} \Delta \dot{\mu}^{(0,1)}(\bm{X}) \right] \\
&=& E\left[ \frac{1}{\pi_{(1,0)}} \left\{ \frac{\pi_{(1,0)}(\bm{X}) E(I(g(\bm{A}) = (0,1) )  \mid \bm{X} ) }{\pi_{(0,1)}(\bm{X})} - E(I(A = 1) \mid \bm{X} )  \right\} \Delta \dot{\mu}^{(0,1)}(\bm{X}) \right] \\
&=&  E\left[ \frac{1}{\pi_{(1,0)}} \left\{ \frac{\pi_{(1,0)}(\bm{X}) \pi_{(0,1)}(\bm{X}) }{\pi_{(0,1)}(\bm{X})} - \pi_{(1,0)}(\bm{X})  \right\} \Delta \dot{\mu}^{(0,1)}(\bm{X}) \right] \\
&=& 0.
\end{eqnarray*}
Similarly, we can show that $E[ (w_{2} (\bm{A}, \bm{X}) - w_{0}(\bm{A}) ) \Delta \dot{\mu}^{(0,0)}(\bm{X}) ] = 0$.
\begin{eqnarray*}
 && E[ \dot{w}_{1}(\bm{A}, \bm{X}; \alpha^{0}) \cdot \{\Delta Y - \Delta \mu^{(0,1)}(\bm{X}; \beta^{0}) - E( w_{1}(\bm{A}, \bm{X}; \alpha^{0}) ( \Delta Y - \Delta \mu^{(0,1)}(\bm{X}; \beta^{0})) ) \}  ] \\ 
 &=& E[ \dot{w}_{1}(\bm{A}, \bm{X}; \alpha^{0}) \cdot \{\Delta Y - \Delta \mu^{(0,1)}(\bm{X}; \beta^{0}) \}  ] \\
 &=& E\left[ \frac{I(g(\bm{A}) = (0,1))}{\pi_{(1,0)}} \{ \Delta Y - \Delta \mu^{(0,1)}(\bm{X}; \beta^{0}) \}   \dot{\pi}_{(1,0), (0,1)}(\bm{X}; \alpha^{0})  \right] \\
 &=&  \frac{\pi_{(0,1)}}{\pi_{(1,0)}} E[  (\Delta Y - \Delta \mu^{(0,1)}(\bm{X}; \beta^{0})) \dot{\pi}_{(1,0), (0,1)}(\bm{X}; \alpha^{0})  \mid g(\bm{A}) = (0,1)] \\ 
 &=& \frac{\pi_{(0,1)}}{\pi_{(1,0)}} E[ \{ E(\Delta Y \mid g(\bm{A}) = (0,1), \bm{X}) - \Delta \mu^{(0,1)}(\bm{X}; \beta^{0}) \}  \dot{\pi}_{(1,0), (0,1)}(\bm{X}; \alpha^{0})  ] \\ 
 &=& 0,
\end{eqnarray*}
where $\dot{\pi}_{(1,0), (0,1)}(\bm{X}_{i}; \alpha^{0}) =  \{\pi^2_{(0,1)} (\bm{X}_{i}; \alpha^{0}) \}^{-1} \{\dot{\pi}_{(1,0)}(\bm{X}_{i}; \alpha^{0}) \pi_{(0,1)}(\bm{X}_{i}; \alpha^{0}) - \pi_{(1,0)}(\bm{X}_{i};\alpha^{0}) \dot{\pi}_{(0,1)}(\bm{X}_{i}; \alpha^{0}) \}$.
The first equality follows from Equation~\eqref{eq:aux1}.
Analogously, $E[  \dot{w}_{2}(\bm{A}, \bm{X}; \alpha^{0}) \{\Delta Y - \Delta \mu^{(0,0)}(\bm{X}; \beta^{0}) - E( w_{2}(\bm{A}, \bm{X}; \alpha^{0}) ( \Delta Y - \Delta \mu^{(0,0)}(\bm{X}; \beta^{0})) ) \} ] = 0$. 
Therefore, the terms that involve the derivatives are zeros. Because
\begin{eqnarray*}
E\left[ \frac{I(A = 1)}{\pi_{(1,0)}} \{ \Delta \mu^{(0,1)}(\bm{X}) - \Delta \mu^{(0,0)}(\bm{X}) \} \right] = \delta,
\end{eqnarray*}
we can conclude the following:
\begin{eqnarray*}
&& \sqrt{n}(\hat{\delta}^{dr} - \delta) \\
&=& \frac{1}{\sqrt{n}} \sum\limits_{i=1}^{n} w_{1}(\bm{A}_{i}, \bm{X}_{i}) (\Delta Y_{i} - \Delta \mu^{(0,1)}(\bm{X}_{i}) ) - w_{2}(\bm{A}_{i}, \bm{X}_{i}) (\Delta Y_{i} - \Delta \mu^{(0,0)}(\bm{X}_{i}) ) \\
&& \quad + w_{0}(\bm{A}_{i}) (  \Delta \mu^{(0,1)}(\bm{X}_{i}) - \Delta \mu^{(0,0)}(\bm{X}_{i}) ) 
- w_{0}(\bm{A}_{i})  \delta + o_{P}(1) \\ 
&=& \frac{1}{\sqrt{n}} \sum\limits_{i=1}^{n} \frac{\pi_{(1,0)}(\bm{X}_{i} )}{\pi_{(1,0)}} \frac{I(g_{i}(\bm{A}) = (0,1)) }{\pi_{(0,1)}(\bm{X}_{i})} \left( \Delta Y_{i} - \Delta \mu^{(0,1)}(\bm{X}_{i}) \right) + \frac{I(A_{i}=1)}{\pi_{(1,0)}}\Delta \mu^{(0,1)}(\bm{X}_{i})  \\  && - \left\{
 \frac{\pi_{(1,0)}(\bm{X}_{i} )}{\pi_{(1,0)}} \frac{I(g_{i}(\bm{A}) = (0,0)) }{\pi_{(0,0)}(\bm{X}_{i})} \left( \Delta Y_{i} - \Delta \mu^{(0,0)}(\bm{X}_{i}) \right) + \frac{I(A_{i}=1)}{\pi_{(1,0)}}\Delta \mu^{(0,0)}(\bm{X}_{i})\right\} \\ && - \frac{I(A_{i}=1)}{\pi_{(1,0)}}\delta  + o_{P}(1) \\ 
 &=& \frac{1}{\sqrt{n}} \sum\limits_{i=1}^{n} \psi_{\delta,i} (Y_{1}, Y_{0}, \bm{A}, \bm{X}) + o_{P}(1).
\end{eqnarray*}
Then by the Lindeberg-L\'{e}vy Central Limit Theorem, this concludes the asymptotic normality of $\sqrt{n}(\hat{\delta}^{dr} - \delta)$ with asymptotic variance of $\sigma^{2}_{\delta_{sb}}$.

\end{proof}

\subsection{Proof of Corollary~\ref{cor:AOTT}}
\begin{proof}[Proof of Corollary~\ref{cor:AOTT}]

The proof of consistency relies on Theorem 1(a) in~\cite{sant2020doubly} and Theorem~\ref{thm:dr} in our main text, followed by an application of Slutsky's theorem. \cite{sant2020doubly} provided an efficient influence function of the ATT under identificaiton assumptions. 
Let us denote the efficient influence of the ATT by $\psi_{\text{ATT}, i}$, which is given by:
\begin{eqnarray*}
\psi_{\text{ATT}, i}(Y_{1}, Y_{0},  \bm{A}, \bm{X}) &=&  \left( \frac{I(A_{i} = 1)}{\pi_{(1,0)}} - \frac{ \pi_{(1,0)}(\bm{X}_{i}) I(g_{i}(\bm{A}) = (0, 1) )  }{\pi_{(1,0)} \pi_{(0,1)}(\bm{X}_{i})} \right) \left( \Delta Y_{i} - \Delta \mu^{(0,0)}(\bm{X}_{i}) \right)
\\ && - \frac{I(A_{i} = 1)}{ \pi_{(1,0)}} E(Y_{1}-Y_{0} \mid A_{i} = 1).
\end{eqnarray*}

By Equation~\eqref{eq:AOTT}, when both the propensity score working model and the outcome working model for the two control groups are correctly specified, we have the asymptotic linear representation of $\sqrt{n}(\hat{\tau}^{dr} - \tau)$ given by:
\begin{eqnarray*}
&& \sqrt{n}(\hat{\delta}^{dr} - \delta) \\
 &=& \frac{1}{\sqrt{n}} \sum\limits_{i=1}^{n} \psi_{\text{ATT}, i} (Y_{1}, Y_{0}, \bm{A}, \bm{X}) + \psi_{\delta, i} (Y_{1}, Y_{0}, \bm{A}, \bm{X}) + o_{P}(1),
\end{eqnarray*}
where the efficient influence function of $E(Y^{(1,1)}_{1} - Y^{(0,0)}_{1} \mid A = 1)$ is given by $\psi_{\tau, i}(Y_{1}, Y_{0}, \bm{A}, \bm{X}) := \psi_{\text{ATT}, i} (Y_{1}, Y_{0}, \bm{A}, \bm{X}) + \psi_{\delta, i} (Y_{1}, Y_{0}, \bm{A}, \bm{X})$.
\end{proof}

\section{Generalization to continuous neighborhood intervention}
\label{sec:continuous}

We can extend binary neighborhood treatment assignments to a continuous treatment setting.
Consider the exposure mapping $h_{it}(\bm{Z}_{t})$ as continuous, denoting the \textit{level} of treatment received among the neighborhood, e.g., $h_{it}(\bm{Z}_{t}) = 0.3$ indicating 30\% of the neighbors surrounding unit $i$ are treated. 
This definition allows flexibility in the specification of the treatment assignment based on the context of the policy evaluation. For example, one may define $h_{it}(\bm{Z}_{t})$ as the proportion of treated units within a $5$-mile radius at $t>0$.
Using this notation, we can generalize the AOTT in Equation~\eqref{eq:AOTT} in the main text to the following $\text{ATT}(\rho)$ with different levels of $\rho$ ($0 \leq \rho \leq 1$):
\begin{eqnarray*}
\label{eq:ATIT_general}
\text{ATT}(\rho) &:=& E ( Y^{(1, \rho)}_{1} - Y^{(0,0)}_{1} \mid A=1 ).
\end{eqnarray*}
Note that $\text{ATT}(0)$ is the ATT in the presence of spillover when surrounded only by control neighbors (Equation~\eqref{eq:ATT} in the main text) and $\text{ATT}(1)$ is the AOTT. 

\begin{figure}[H]
\centering
\begin{subfigure}[b]{0.28\textwidth}
\includegraphics[width=\textwidth]{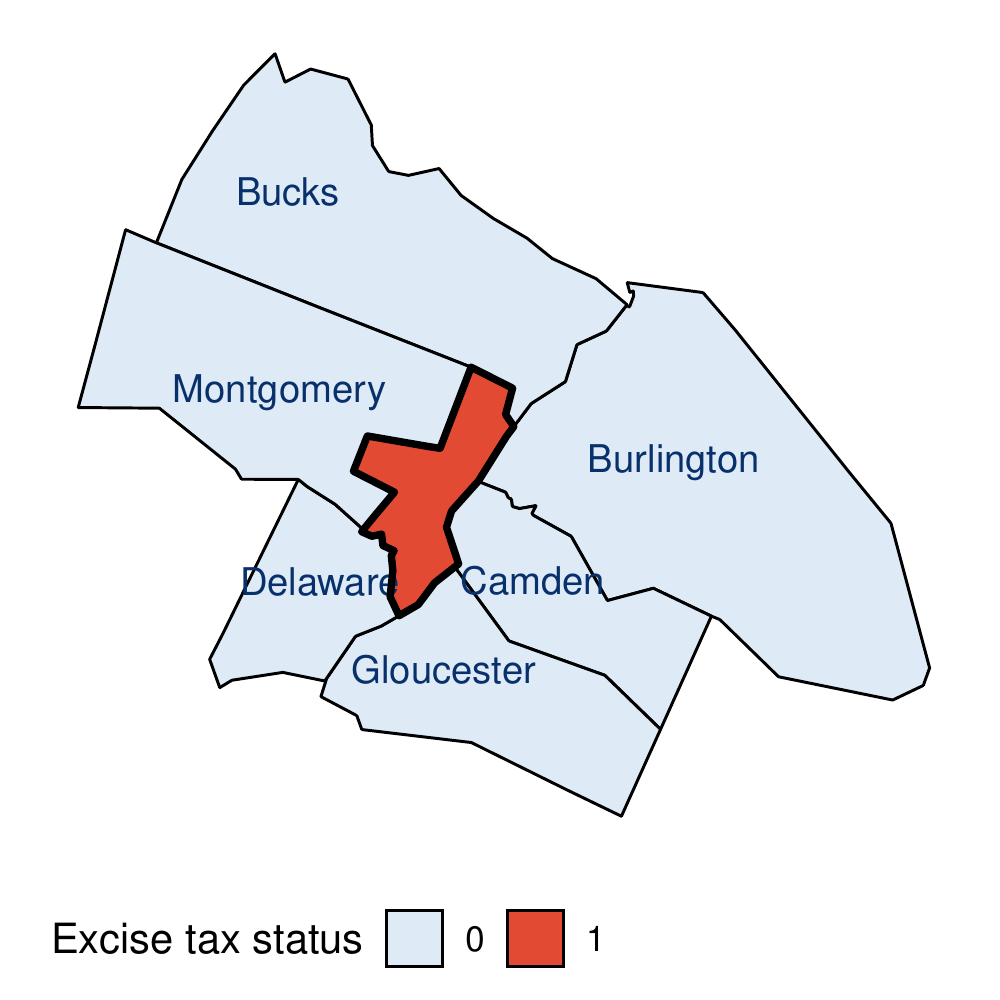}
\caption{$\rho = 0.0$}
\end{subfigure}
\begin{subfigure}[b]{0.28\textwidth}
\includegraphics[width=\textwidth]{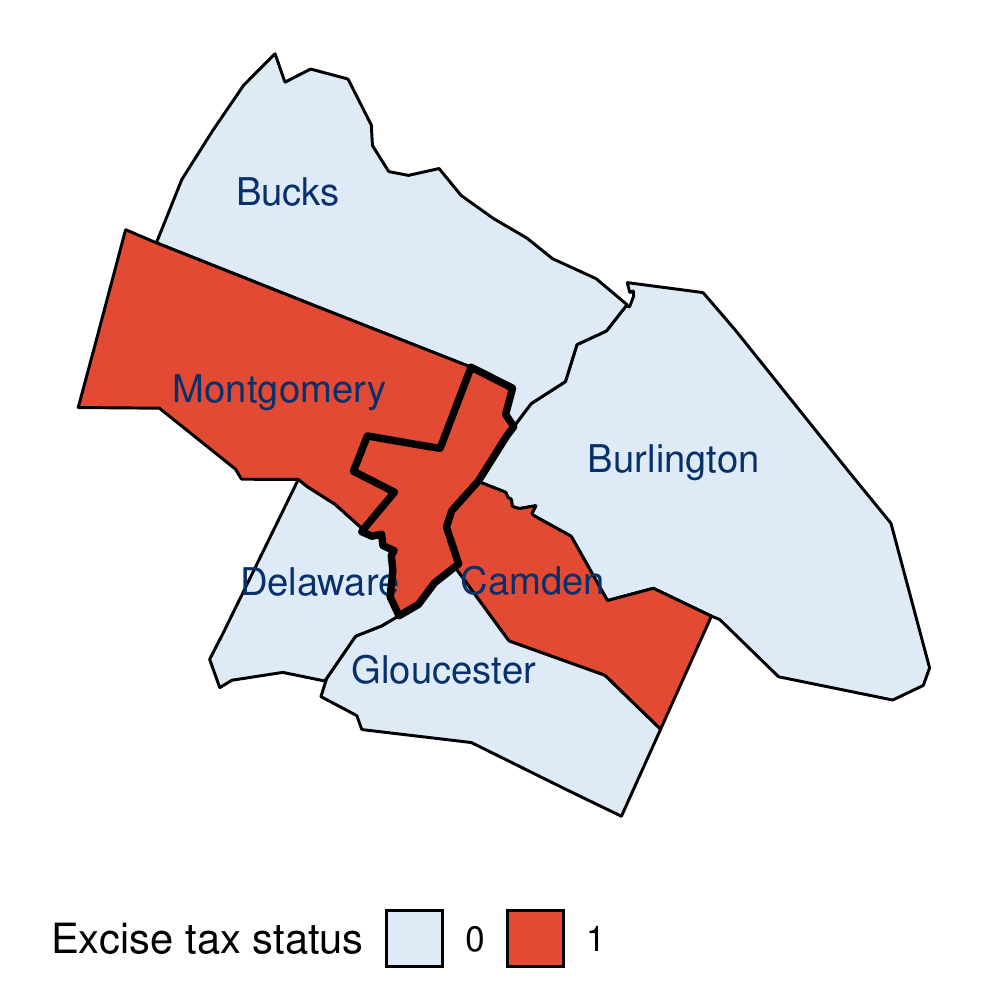}
\caption{$\rho = 0.4$}
\end{subfigure}
\begin{subfigure}[b]{0.28\textwidth}
\includegraphics[width=\textwidth]{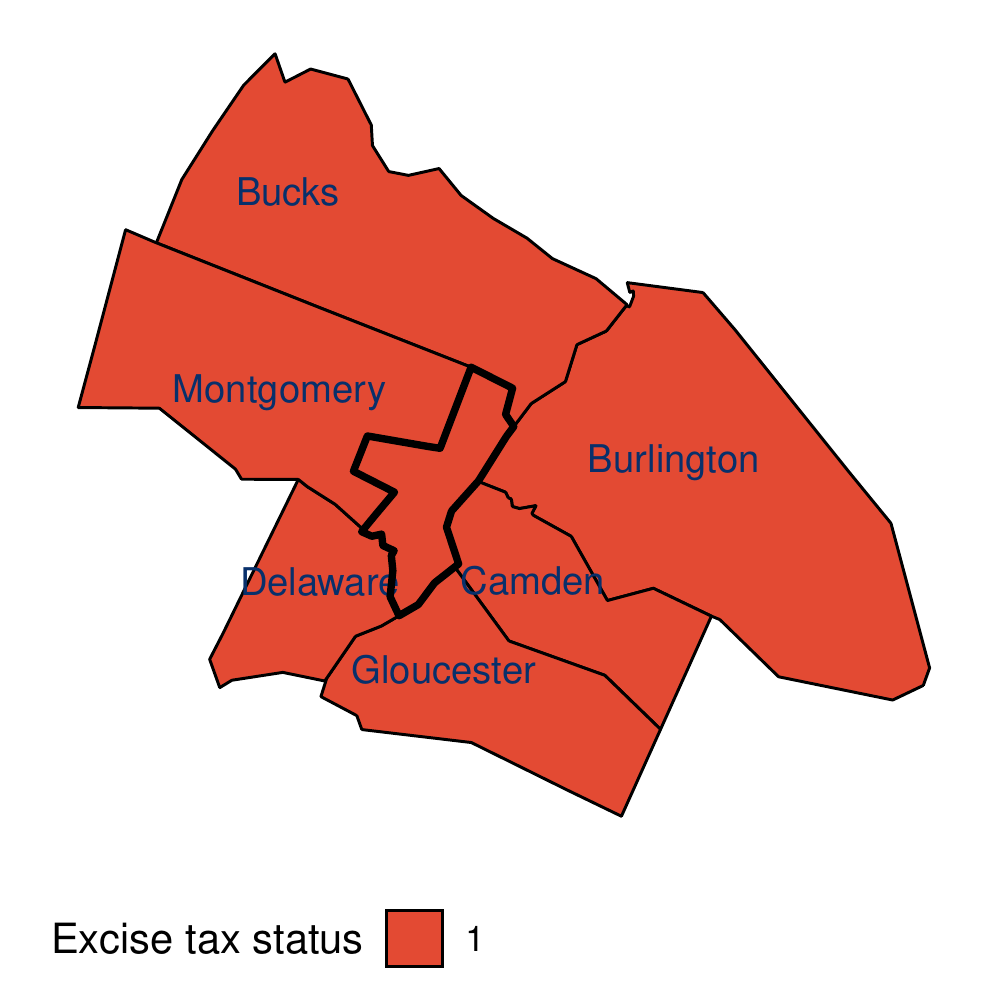}
\caption{$\rho = 1.0$}
\end{subfigure}
\caption{\label{fig:usmap} Hypothetical examples of varying proportions of treated neighborhoods surrounding Philadelphia}
\end{figure}

Figure~\ref{fig:usmap} illustrates one observed and two counterfactual neighborhood treatment assignment (e.g., excise tax) scenarios for Philadelphia and its neighboring counties.  Depending on the question of interest, policy makers may vary the neighborhood treatment assignment levels from $\rho = 0.0$ to $1.0$ and evaluate the treatment effect on Philadelphia under each level.
For example, suppose that a treated group had $(100 \times \rho)\%$ of its neighbors treated rather than 0\%. Then researchers may assume that $(100 \times \rho) \%$ of the spillover effects on neighboring controls would have been returned to the treated group if the baseline covariates $\bm{X}$ are similar between the two groups. In other words, the following counterfactual equation holds true. 
\begin{eqnarray*}
\label{eq:travel}
E(Y^{(1,0)}_{1} - Y^{(1,\rho)}_{1} \mid A = 1, \bm{X}) = 
\rho \times E(Y^{(0,1)}_{1} - Y^{(0,0)}_{1} \mid g(\bm{A}) = (0,1), \bm{X}_{i}).
\end{eqnarray*}
Then $\text{ATT}(\rho) = E(Y^{(1,0)}_{1} - Y^{(0,0)}_{1} \mid A = 1)- E(\rho \times E(Y^{(0,1)}_{1} - Y^{(0,0)}_{1} \mid g(\bm{A}) = (0,1), \bm{X})\mid A = 1)$.
As we can identify the conditional spillover effects as well as the ATT under the aforementioned assumptions, the above generalized ATT can also be identifiable.

\section{Additional Simulations}
\label{sec:additional}

In this section, we compare the finite-sample performance of our doubly robust estimator to the inverse probability weighted (IPW) and the outcome regression-based (Reg) estimators for four different estimands: ATT, ATN, offsetting effect, and AOTT. 
We consider three different scenarios: (a), (b), and (c) that are introduced in the main text.
In case of (b) the incorrect outcome regression model, we omitted the interaction terms between the time indicator and the baseline covariates $\bm{X}_{i}$, and between the neighborhood treatment assignment and $X^{2}_{i2}$. In the (c) misspecified propensity score model, we omitted $X_{i1}$ and included $\exp(X^{2}_{i2})$ instead of $X_{i2}$. Other settings remain the same as those in the main text.

\begin{table}[!ht]
  \centering
\resizebox{0.65\textwidth}{!}{\begin{tabular}{r|rrrr}
    & ATT & ATN & Offsetting effect & AOTT \\ 
    \hline
    \multicolumn{5}{l}{ (a) all correct} \\
    \hline
IPW & 0.00 (94.0) & -0.01 (95.6) & 0.01 (94.9) & 0.01 (94.7) \\ 
  Reg & 0.00 (95.2) & -0.01 (95.7) & -0.00 (94.9) & 0.00 (95.4) \\ 
  DR & 0.00 (95.1)& -0.01 (95.6)&  -0.00 (95.4)&  -0.00 (95.0)\\ 
    \hline
    \multicolumn{5}{l}{ (b) Incorrect outcome regression} \\
    \hline  
  IPW & 0.00 (93.9) & -0.01 (96.0) & 0.01 (94.9) & 0.01 (94.7) \\ 
  Reg & -0.50 (0.00) & 0.17 (59.8) & 0.31 (45.7) & -0.19 (84.9) \\ 
  DR & 0.00 (93.8) & -0.01 (95.5) &  0.02 (94.8) &  0.03 (95.2) \\ 
     \hline
     \multicolumn{5}{l}{ (c) Incorrect propensity scores} \\
     \hline
IPW & -0.49 (0.00) & 0.17 (55.6) & -0.10 (82.0) & -0.59 (2.9)\\ 
  Reg & 0.00 (95.7) & 0.00 (95.3) & 0.00 (95.3) & 0.01 (95.1) \\ 
  DR & 0.00 (94.8) & -0.00 (95.0) & 0.00 (95.1) &  0.01 (95.1) \\ 
     \hline
  \end{tabular}}
\caption{ \label{tab:all} Bias (a coverage rate of 95\% confidence intervals) based on $1000$ independent replicates when $n=2000$. The coverage rates are calculated based on the empirical standard errors of the estimates from $1000$ replicates.}
\end{table} 

Table~\ref{tab:all} presents the bias and coverage rates of each estimator for the four estimands.
When (a) both of the outcome regression model $\{\mu^{(a_{1}, a_{2})}_{t}(\bm{X})\}$ and the propensity score model $\{ \pi_{(a_{1}, a_{2})}(\bm{X}) \}$ are correctly specified, as expected, all three estimators achieve a nominal coverage rate and their biases are all near zero. When (b) the outcome regression model is misspecified, the outcome regression-based estimator performs poorly across all four estimands whereas the IPW estimator is not affected. When (c) the propensity scores are not correctly specified, the bias of the IPW estimator is high and its coverages are far lower than the nominal rate.
On the other hand, across all scenarios, the doubly robust estimator achieves a near-zero bias and nominal coverage rate.

\section{Simulation with multiple time points}
\label{sec:multiple}

In this section, we demonstrate the performance of our proposed estimator for time averaged treatment effects (Equation~\eqref{eq:time_AOTT} in the main text) when we have $T$ pre-treatment time periods (i.e., $t=-(T-1), -(T-2), \ldots, 0$) and $T$ post-treatment time periods (i.e., $t=1,2,\ldots, T$).  
We consider the following potential outcomes given $g_{i}(\bm{A}) = (a_{1}, a_{2})$ and $\bm{X}_{i}  = (X_{i1}, X_{i2})$ under no policy intervention for $t$=-12,-11, $\cdots$, 0, 1, $\cdots$, 13:
\begin{eqnarray*}
 Y^{(0,0)}_{it}  &=& \gamma_{t} + X_{1i}\lambda_{1t} + X_{2i}\lambda_{2t} + \alpha_{(a_{1},a_{2})} + \epsilon_{it}.
\end{eqnarray*}
We specify baseline group-specific effects via $(\alpha_{(0,0)}, \alpha_{(1,0)}, \alpha_{(0,1)}) = (0.0, 0.5, -1.0)$. 
Here we generate error terms from a bivariate normal at each pair of time points $(t-T,~t)$ independently across $t=1,2,\ldots, T$.
We set $\boldsymbol{\gamma} = (-1.0. -0.9, \ldots, 1.5)$ as time-specific effects. We specify the time-varying effects of the two covariates by
$\boldsymbol{\lambda}^{T}_{1,-12:0}$ = (-0.3, -0.2, \ldots, 0.3, 0.2 \ldots, -0.3), $\boldsymbol{\lambda}^{T}_{1, 0:13}$ = (-0.2, -0.3, \ldots, 0.4, 0.3 \ldots, -0.2)), $\boldsymbol{\lambda}^{T}_{2, -12:0}$ = (0.3, 0.2, \ldots, -0.3, -0.2 \ldots, 0.3), and  $\boldsymbol{\lambda}^{T}_{2, 1:13}$ = (0.0, -0.1, \ldots, -0.6, -0.5 \ldots, 0.0)). 
We then generate the following potential outcomes across $t=-12,-11, 0, 1, \ldots, 13$ for control units with $g_{i}(\bm{A}) = (0,1)$ and the treated units with $g_{i}(\bm{A}) = (1,0)$:
\begin{eqnarray*}
 Y^{(0,1)}_{it}  &=& Y^{(0,0)}_{it} +  \theta_{(0,1)} \eta_{(0,1),t} \{( 1+0.5X_{i1} ) + (1-X^2_{i2}) \} \\ 
 Y^{(1,0)}_{it} &=&  Y^{(0,0)}_{it} + \theta_{(1,0)} \eta_{(1,0),t} ( 1-0.5X_{i1}).
\end{eqnarray*}
We consider $(\theta_{(0,1)}, \theta_{(1,0)})$ = (1.0, -1.5) and (0.5, -2.0), so that the intervention has a negative effect on the treatment group's outcomes while it has a positive effect on the neighboring control's outcomes, on average. However, both of these effects may vary depending on covariates $X_{i1}$ and $X_{i2}$ and time points.
The following parameters specify the time-varying treatment effects: 
$\eta_{(0,1), t} = 0$ if $t \leq 0$ (i.e., no spillover effects at the time or prior to the implementation of the policy intervention) and $\eta_{(0,1), t} = 1 - (t-1)/12$ for $t > 0$ (i.e., decreasing spillover effects over time after the implementation of the policy intervention).  
Similarly,  we set $\eta_{(1,0), t} = 0$ if $t \leq 0$, $\eta_{(1,0), 1} = 0.8,~\eta_{(1,0), 2} = 0.9$, and $\eta_{(1,0), t} = 1 - (t-1)/10$ for $t \geq 3$, reflecting the time-lag to reach the maximized effect of the policy intervention on the treatment group.

Under these settings, we derive the doubly robust estimator for the AOTT, $\tau^{\text{dr}}_{t}$ at each time point $t=1,\ldots,13$. We then aggregate the causal estimates to have a \textit{summary} causal measure using DiD with multiple time periods. 
To estimate the variance, we can use $\text{var}(T^{-1}\sum_{t=1}^{T} n^{-1} \sum_{i=1}^{n} \hat{\psi}_{\tau,i, t})$, where $\psi_{\tau, i, t}$ is the influence function for $\hat{\tau}^{\text{dr}}_{t}$, assuming that the estimators at each time point are independent. Table~\ref{tab:multiresults} presents the simulation results for the time average AOTT using the doubly-robust estimator (Equation~\eqref{eq:time_AOTT}) under the three different nuisance model specifications scenarios when $(\theta_{(0,1)}, \theta_{(1,0)})$ is (1.0, -1.5) and (0.5, -2.0), respectively. In case of (b) the incorrect outcome regression, we have omitted the interaction terms between the time indicator and the baseline covariates $\bm{X}_{i}$ and the interaction between the neighborhood treatment assignment and $X^{2}_{i2}$. In the misspecified propensity score model, we have omitted $X_{i1}$ and included $\exp(X_{i2})$ instead of $X_{i2}$. 
We observe the the standard errors of the estimates are smaller when (a) all of the propensity scores and the outcome regression are correctly specified compared to those when either one of these models is misspecified. However, the bias in all cases converges to zero as the sample size increases. The coverage rates of 95\% confidence intervals are close to the nominal rate. 

\begin{table}[!ht]
\centering
\resizebox{\textwidth}{!}{\begin{tabular}{l|ccc|ccc|ccc}
Case & \multicolumn{3}{c|}{(a) All correct} & \multicolumn{3}{c|}{(b) Incorrect outcome regression} &  \multicolumn{3}{c}{(c) Incorrect propensity scores} \\
 & Bias & SE & CR & Bias & SE & CR & Bias & SE & CR \\
 \hline
\multicolumn{10}{l}{(1) $\theta_{(0,1)} = 1.0,~\theta_{(1,0)} = -1.5$} \\
\hline 
 $n=500$ & 0.019 & 0.447 & 0.934 & 0.032 & 0.500 & 0.966 & 0.058 & 0.630 & 0.931\\
 $n=1000$ &  0.003 & 0.325 & 0.935 & 0.006 & 0.355 & 0.957 & 0.001 & 0.444 & 0.953 \\ 
$n=2000$  & -0.006 &  0.233 & 0.948 & 0.002 & 0.254 & 0.959 & -0.001 & 0.333 &  0.960 \\
\hline
\multicolumn{10}{l}{(1) $\theta_{(0,1)} = 0.5,~\theta_{(1,0)} = -2.0$} \\
\hline
$n=500$ & 0.032 & 0.444 & 0.936 & 0.028 & 0.469 & 0.956 & 0.054 & 0.628 & 0.935\\
$n=1000$ & 0.003 & 0.323 & 0.936 &  0.006 & 0.334 & 0.951 & 0.002 & 0.442 & 0.952\\
$n=2000$ & 0.003 & 0.232 & 0.944 & 0.002 & 0.238 & 0.954 & -0.004 &  0.331 & 0.960\\
\hline
\end{tabular}}
\caption{\label{tab:multiresults} Simulation results of the time-average AOTT (Equation~\eqref{eq:time_AOTT}) with $T = 13$ when (a) all of the propensity scores and the outcome regression are correctly specified, (b) the outcome regression is misspecified, and (c) the propensity scores are misspecified. Standard errors (SE) and coverage rates (CR) of 95\% confidence intervals are based on the influence function-based variance estimator. }
\end{table}

Table~\ref{tab:all2} shows the bias and the coverage rate of 95\% confidence intervals for four different time average estimands: ATT, ATN, offsetting effect, and the AOTT. Here we compare the performance of the IPW and the outcome regression-based (Reg) estimators to the doubly-robust (DR) estimator. Similar to the case with $T=1$, all three estimators have biases that are close to zero and have coverage rates of 95\% confidence intervals that are close to 95\% for all of the estimands when (a) both of the outcome regression model and the propensity score model are correctly specified. 
However, the coverage rates of the outcome regression-based estimator and the IPW estimator significantly decline to below 95\% under scenario (b) and (c), respectively, especially for the offsetting effect and the AOTT.

\begin{table}[!ht]
  \centering
\resizebox{0.65\textwidth}{!}{\begin{tabular}{l|rrrr}
    & ATT & ATN & Offsetting effect & AOTT \\ 
    \hline
    \multicolumn{5}{l}{ (a) all correct} \\
    \hline
  IPW & -0.00 (94.9) & -0.00 (96.0) & -0.01 (95.0) & -0.01 (94.7) \\ 
  Reg & -0.00 (95.3) & -0.00 (96.2) & -0.00 (94.4) & -0.00 (95.0) \\ 
  DR & -0.00 (94.6) & -0.00 (96.1)& -0.00 (95.8)& -0.01 (94.8) \\ 
    \hline
    \multicolumn{5}{l}{ (b) Incorrect outcome regression} \\
    \hline  
  IPW & 0.00 (94.9) & -0.01 (95.8) & -0.00 (94.9) & -0.00 (94.8) \\ 
  Reg & 0.00 (95.0) & -0.01 (95.4) & -0.29 (3.4) & -0.29 (39.9) \\ 
  DR & 0.00 (94.9) & -0.01 (95.7) & -0.00 (94.5) & 0.00 (94.9) \\ 
     \hline
     \multicolumn{5}{l}{ (c) Incorrect propensity scores} \\
     \hline
IPW & 0.00 (94.9) & -0.01 (94.8) & -0.87 (55.5) & -0.87 (58.5) \\ 
  Reg & 0.00 (95.4) & -0.01 (95.5) & 0.00 (94.5) & 0.00 (95.0) \\ 
  DR & 0.00 (95.3) & -0.01 (95.8) & -0.01 (94.5) & -0.00 (94.5) \\ 
 \hline
  \end{tabular}}
\caption{ \label{tab:all2} Bias (a coverage rate of 95\% confidence intervals) based on $1000$ independent replicates when $n=2000$. The coverage rates are calculated based on the empirical standard errors of the estimates from $1000$ replicates when (1) $\theta_{(0,1)} = 1.0$ and $\theta_{(1,0)} = -1.5$.}
\end{table} 

\newpage
\section{Further Results: Philadelphia Beverage Tax Study}
\label{sec:D}

Here, we provide further results of our analysis of the Philadelphia beverage tax data. Table~\ref{tab:propensity_result} presents the estimated standardized coefficients and their p-values from the multinomial propensity score model based on 13, 15, and 26 supermarkets in Baltimore, neighboring counties of Philadelphia, and Philadelphia, respectively.  

\begin{table}[H]
\centering
\resizebox{\textwidth}{!}{\begin{tabular}{lrrrrr}
  \hline
 & Intercept & Family SS & Family AS & Individual SS & Individual AS \\
  \hline
PH & 6.548 ($<$0.001) & -6.594 ($<$0.001) & 7.278 ($<$0.001) & 7.140 ($<$0.001) & -0.383  (0.702) \\ 
NC   & -0.112 (0.911) & -3.412 (0.001) & 2.754 (0.006) & 9.315 ($<$0.001) & -4.105 ($<$0.001) \\ 
   \hline
  & Black  & Hispanic & Asian & Male & Income per household \\ 
  \hline
PH &  -4.497 ($<$0.001) & 4.183 ($<$0.001) & 4.666 ($<$0.001) & -6.629  ($<$0.001)& -1.145 (0.252) \\ 
NC & -4.322 ($<$0.001) & -3.511 ($<$0.001) & 5.970 ($<$0.001) & 0.292 (0.770) & -4.616 ($<$0.001)  \\   
  \hline
\end{tabular}}
\caption{\label{tab:propensity_result} Standardized coefficients ($p$-value) from the multinomial propensity score model with Baltimore as a reference group (PH: Philadelphia, NC: neighboring counties of Philadelphia, SS: sugar-sweetened beverages, AS: artificially sweetened beverages)}
\end{table}

Table~\ref{tab:outcome_result} presents the estimated standardized coefficients and their p-values from the outcome regression model for two different outcomes: the total unit sales of (1) the taxed individual size beverages and (2) the taxed family size beverages in supermarkets. 

\begin{table}[H]
\centering
\resizebox{\textwidth}{!}{\begin{tabular}{l|cc|cc}
\multicolumn{1}{r|}{\textbf{Outcome variables:}} & \multicolumn{2}{c|}{(1) Taxed individual sized beverages}  & \multicolumn{2}{c}{(2) Taxed family sized beverages}   \\
  \hline
\textbf{Exposure variables} & $t$-value & $p$-value & $t$-value & $p$-value \\ 
\hline
Intercept &   1.501 & 0.133 & 1.355 & 0.175 \\ 
PH & -1.195 & 0.232 & 2.532 & 0.011 \\  
NC & -1.695 & 0.090  & 0.655 & 0.513 \\ 
Year 2017 & -3.211 & 0.001  & 0.057 & 0.955 \\ 
PH $\times$ Year 2017 & -7.043 & $<0.001$  & -16.528 & $<0.001$  \\ 
NC $\times$ Year 2017 & 0.824 & 0.410 & -0.365 & 0.715 \\ 
Individual SS & 0.681 & 0.496 &  -1.067 & 0.286 \\
Family SS & -0.896 & 0.370 & 0.342 & 0.732 \\ 
Individual AS & -2.437 & 0.015& -3.391 & 0.001 \\ 
Family AS & 2.042 & 0.041 &  1.382 & 0.167 \\ 
Black &  -0.505 & 0.614 & -0.219 & 0.826 \\ 
Hispanic & 1.063 & 0.288 & 0.885 & 0.376 \\ 
Asian & -0.324 & 0.746 & -0.512 & 0.609 \\ 
Male & -0.967 & 0.334  & -1.096 & 0.273 \\ 
Income per household & 0.363 & 0.716   & 0.614 & 0.539 \\ 
PH $\times$ Year 2017 $\times$ Individual SS & 1.855 & 0.064 & 7.589 & $<0.001$  \\ 
NC $\times$ Year 2017 $\times$ Individual SS & -1.351 & 0.177  & 0.629 & 0.529 \\
PH $\times$ Year 2017 $\times$ Family SS &  7.426 & $<0.001$   & 12.421 & $<0.001$  \\ 
NC $\times$ Year 2017 $\times$ Family SS & 2.452 & 0.014  &  0.986 & 0.324 \\
PH $\times$ Year 2017 $\times$ Individual AS & -2.615 & 0.009  &  -6.003 & $<0.001$  \\ 
NC $\times$ Year 2017 $\times$ Individual AS & -2.111 & 0.035 & -0.690 & 0.490 \\  
PH $\times$ Year 2017 $\times$ Family AS & -8.742 & $<0.001$  & -13.874 & $<0.001$  \\ 
NC $\times$ Year 2017 $\times$ Family AS & -2.932 & 0.003  & -1.419 & 0.156 \\ 
 \hline
\end{tabular}}
\caption{\label{tab:outcome_result} Standardized coefficients ($p$-value) from the outcome regression model with Baltimore as a reference group (PH: Philadelphia, NC: neighboring counties, SS: sugar-sweetened beverages, AS: artificially sweetened beverages)}
\end{table}

\end{document}